\documentclass[11pt]{myclass}
\usepackage[margin=1in]{geometry}
\usepackage{mathptmx}
\usepackage{color}
\usepackage{hyperref}
\usepackage{verbatim}
\usepackage{amssymb}
\usepackage{float,ctable}
\usepackage{algorithm, algorithmic}
\usepackage{wrapfig}
\newcommand{\breg}{\ensuremath{D_\phi}}
\newcommand{\sbreg}{\ensuremath{D_{s\phi}}}

\newcommand{\eps}{\varepsilon}
\newcommand{\etal}{\emph{et al}\xspace}

\title{Approximate Bregman near neighbors in sublinear time: Beyond the triangle inequality}
\author{Amirali Abdullah\\University of Utah \and John Moeller\\University of Utah \and Suresh Venkatasubramanian\\University of Utah}
\date{}
\begin{document}
\begin{titlepage}
\maketitle
\thispagestyle{empty}
\begin{abstract}
\emph{Bregman divergences} are  important distance measures that are used extensively in data-driven applications such as computer vision, text mining, and speech processing, and are a key focus of interest in machine learning. Answering \emph{nearest neighbor} (NN) queries under these measures is very important in these applications and has been the subject of extensive study, but is problematic because these distance measures  lack metric properties like symmetry and the triangle inequality.
In this paper, we present the first provably  \emph{approximate nearest-neighbor} (ANN)  algorithms for a broad sub-class of Bregman divergences under some assumptions. Specifically, we examine Bregman divergences which can be decomposed along each dimension and our bounds also depend on restricting the size of our allowed domain. We obtain bounds for both the regular asymmetric Bregman divergences as well as their symmetrized versions. 
To do so, we develop two geometric properties vital to our analysis: a \emph{reverse triangle inequality} 
(RTI) and a relaxed triangle inequality called \emph{$\mu$-defectiveness} where $\mu$ is a domain-dependent value. 
Bregman divergences  satisfy the RTI but \emph{not} $\mu$-defectiveness. However, we show that the square root of a 
Bregman divergence does satisfy $\mu$-defectiveness. This allows us to then utilize both properties in
 an efficient search data structure that follows the general two-stage paradigm of a ring-tree 
decomposition followed by a quad tree search used in previous near-neighbor algorithms for Euclidean space and spaces of bounded doubling dimension.

Our first algorithm resolves a query for a $d$-dimensional $(1+\eps)$-ANN in $O \left(\left(\frac{\mu \log n}{\eps}\right)^{O(d)} \right)$ time and $O \left(n \log^{d-1} n \right)$ space and holds for generic $\mu$-defective distance measures satisfying a RTI. 
Our second algorithm is more specific in analysis to the Bregman divergences and uses a further structural parameter, the maximum ratio of second derivatives over each dimension of our allowed domain ($c_0$). This allows us to locate a $(1+\eps)$-ANN in $O(\log n)$ time and $O(n)$ space, where there is a further $(c_0)^d$ factor in the big-Oh for the query time.

\end{abstract}  
\end{titlepage}

\section{Introduction}
\label{Introduction}
The nearest neighbor problem is one of the most extensively studied problems in data analysis. The past 20 years has seen tremendous research into the problem of computing near neighbors efficiently as well as approximately in different kinds of metric spaces.

An important application of the nearest-neighbor problem is in querying content databases (images, text, and audio databases, for example). In these applications, the notion of similarity is  based on a distance metric that arises from information-theoretic or other considerations. Popular examples include the Kullback-Leibler divergence~\cite{kullback}, the Itakura-Saito distance~\cite{itakura} and the Mahalanobis distance~\cite{mahalanobis36}. These distance measures are examples of a general class of divergences called the \emph{Bregman divergences}~\cite{bregman}, and this class has received much attention in the realm of machine learning, computer vision and other application domains.

Bregman divergences possess a rich geometric structure but are not metrics in general, and are not even symmetric in most cases! While the  geometry of Bregman divergences has been studied from a combinatorial perspective and for clustering, there have been no algorithms with provable guarantees for the fundamental  problem of nearest-neighbor search. This is in contrast with extensive \emph{empirical} study of Bregman-based near-neighbor search\cite{caytonpaper,vptrees,tailoredbregmannn,spellmanvemuri,bregsearch}. 

In this paper we present the first provably approximate nearest-neighbor (ANN) algorithms for a broad sub-class of Bregman divergences, with an assumption of restricted domain. Our first algorithm processes queries in $O(\log^d n)$ time using $O(n \log^d n)$ space and only uses general properties of the underlying distance function (which includes Bregman divergences as a special case). The second algorithm processes queries in $O(\log n)$ time using $O(n)$ space and exploits structural constants associated specifically with Bregman divergences. An interesting feature of our algorithms is that they extend the ``ring-tree + quad-tree'' paradigm for ANN searching beyond Euclidean distances and metrics of bounded doubling dimension to distances that might not even be symmetric or satisfy a triangle inequality. 

\subsection{Overview of Techniques}
\label{ssec:intro-key-ideas}

At a high level\cite{snotes}, low-dimensional Euclidean approximate near-neighbor search  works as follows. The algorithm builds a quad-tree-like data structure to search the space efficiently at query time. Cells reduce exponentially in size, and so a careful application of the triangle inequality and some packing bounds allows us to bound the number of cells explored in terms of the ``spread'' of the point set (the ratio of the maximum to minimum distance). Next, 
% One approach~\cite{snotes} for near-neighbor search in low dimensions is to build a quad-tree-like data structure to process queries. Since the quad tree cells reduce in size by a constant factor at each stage, the triangle inequality can be used to infer that we never need to expand cells that are smaller than a fraction of the true nearest neighbor distance in order to get a good approximation to the nearest neighbor. This fact is then combined with a packing bound that upper bounds the number of such cells we need to explore in order to obtain a query running time that is a function only of the desired error and the . 
terms involving the spread are eliminated by finding an initial crude approximation to the nearest neighbor. 
Since the resulting depth to explore is bounded by the logarithm of the ratio of the cell sizes, any $c$-approximation of the nearest neighbor results in a depth of $O \left(\log(c/\eps) \right)$. A standard data structure that yields such a crude bound is the \emph{ring tree}~\cite{blackbox}.

Unfortunately, these methods (which work also for doubling metrics~\cite{bounded,blackbox}) 
require two key properties: the existence of the triangle inequality, as well as packing bounds for fitting small-radius balls into large-radius balls. Bregman divergences in general are not symmetric and do not even satisfy a directed triangle inequality! We note in passing that such problems do not occur for the \emph{exact} nearest neighbor problem in constant dimension: this problem reduces to point location in a Voronoi diagram, and Bregman Voronoi diagrams possess the same combinatorial structure as Euclidean Voronoi diagrams~\cite{bvd}. The complexity of a Voronoi diagram of $n$ points is well known to be $O(n^{\frac{d}{2}})$, and as such of prohibitive space complexity.

\paragraph*{Reverse Triangle Inequality}
The first observation we make is that while Bregman divergences do not satisfy a triangle inequality, they satisfy a weak \emph{reverse triangle inequality}: 
along a line, the sum of lengths of two contiguous intervals is always \emph{less} than the length of the union. 
This immediately yields a packing bound: intuitively, we cannot pack too many disjoint intervals in a larger interval because their sum would then be too 
large, violating the reverse triangle inequality. 

\paragraph{$\mu$-defectiveness}

The second idea is to allow for a \emph{relaxed} triangle inequality. We do so by defining a distance measure to be 
\emph{$\mu$-defective} w.r.t a given domain if there exists a fixed $\mu \ge 1$ such that for all triples of points $x,y,z$ , we have that
$|D(x,y) - D(x,z)| \le \mu D(y,z)$. This notion was first employed by Farago et.al \cite{firstdefect} for an algorithm based on optimizing average case
complexity. 

A different natural way to relax the triangle inequality would be to show there exists a fixed $\mu < 1$ such that 
for all triples $(x,y,z)$, the inequality $D(x,y) + D(y,z) \ge \mu D(x,z)$. In fact, this is the notion of \emph{$\mu$-similarity} used by 
Ackermann \etal~\cite{musimilarcoresets}  to \emph{cluster} data under a Bregman divergence. However, this version of a relaxed triangle inequality 
is too weak for the nearest-neighbor problem, as we see in Figure\ref{counter}. 

\begin{figure}[H]
  \begin{center}
    \includegraphics[scale = 0.4]{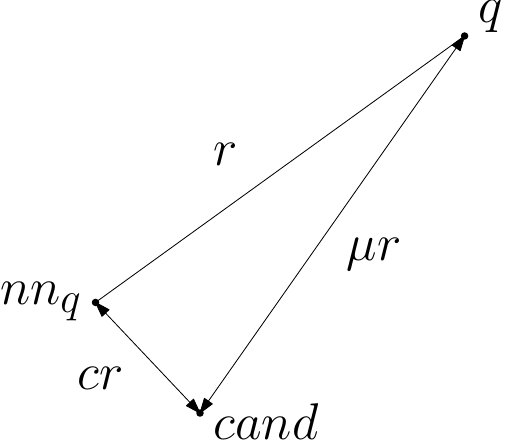}
  \end{center}
  \caption{The ratio $\frac{D(q,\text{cand})}{D(q,\text{nn}_q)} = \mu$, no matter how small $c$ is}
  \label{counter}
\end{figure}

Let $q$ be a query point, $\text{cand}$ be a point from $P$ such that $D(q,\text{cand})$ is known and $\text{nn}_q$ 
be the actual nearest neighbor to $q$. The principle of grid related machinery is that for $D(q, \text{nn}_q)$ and $D(q,\text{cand})$ sufficiently 
large, and $D(\text{cand}, \text{nn}_q)$ sufficiently small, we can verify that $D(q,\text{cand})$ is a $(1+ \eps)$ nearest neighbor, 
i.e we can short-circuit our grid. 

The figure \ref{counter} illustrates a case where this short-circuit may not be valid for $\mu$-similarity. Note that $\mu$-similarity is satisfied 
here for any $c <1$. Yet the ANN quality of $\text{cand}$, i.e, $\frac{D(q,\text{cand}) }{ D(q, \text{nn}_q)}$, need not be better than $\mu$ even 
for arbitrarily close $\text{nn}_q$ and $\text{cand}$! This demonstrates the difficulty of naturally 
adapting the Ackermann notion of $\mu$-similarity to finding a $1 + \eps$ nearest neighbor.

In fact, the relevant relaxation of the triangle inequality that we require is slightly different. Rearranging terms, we instead require that there exist 
a parameter $\mu \ge 1$ such that for all triples $(x,y,z)$, $|D(x,y) - D(x,z)| \le \mu D(y,z)$. We call such a distance \emph{$\mu$-defective}. 
It is fairly straightforward to see that a $\mu$-defective distance measure is also $2/(\mu+1)$-similar, but the converse does not hold, 
as the example above shows.

Without loss of generality, assume that $D(x,y) \geq D(x,z) \geq D(y,z)$.  Then $D(x,y) - D(x,z) \le \mu D(y,z)$ and $D(x,y) - D(y,z) \le \mu D(x,z)$, 
so $2D(x,y) \le (\mu+1)( D(x,z)+D(y,z))$. Since $D(x,y)$ is the greatest of the three distances, this inequality is the strongest and implies 
the corresponding $2/(\mu+1)$-similarity inequalities for the other two distances.

Unfortunately, Bregman divergences do not satisfy $\mu$-defectiveness for any size domain or value of $\mu$! One of our
technical contributions is demonstrating in Section \ref{sec:prop-sqrtd_s-phi} that surprisingly, the square root of Bregman divergences
does satisfy this property over restrictions of our domain with $\mu$ depending on the boundedness of this subdomain we consider and the choice of divergence.

\paragraph{A Generic Approximate Near-Neighbor Algorithm}

After establishing that Bregman divergences satisfy the reverse triangle inequality and $\mu$-defectiveness (Section~\ref{sec:prop-sqrtd_s-phi}),
 we first show (Section~\ref{sec:ringsec}) that \emph{any} distance measure satisfying 
the reverse triangle inequality, $\mu$-defectiveness, and some mild technical conditions admits a ring-tree-based construction 
to obtain a weak near neighbor. However, applying it to a quad-tree construction creates a problem. 
The $\mu$-defectiveness of a distance measure means that if we take a unit length interval and divide it into two parts, 
all we can expect is that each part has length between $1/2$ and $1/(\mu+1)$. This implies that while we may have 
to go down to level $\lceil \log_2 \ell\rceil$ to guarantee that all cells have side length $O(\ell$), some cells 
might have side length as little as $\ell^{\log_2 (\mu+1)}$, weakening packing bounds considerably. 

We deal with this problem in two ways. For Bregman divergences, we can exploit geometric properties of the 
associated convex function $\phi$ (see Section~\ref{sec:defn}) to ensure that cells at a fixed level have bounded
size (Section~\ref{sec:condition}); this is achieved by reexamining the second derivative $\phi''$. 

For more general abstract distances that satisfy the reverse triangle inequality and $\mu$-defectiveness,
 we instead construct a portion of the quad tree ``on the fly'' for each query (Section~\ref{sec:finalized-algorithm}).
 While this is expensive, it still yields polylog($n$) bounds for the overall query time in fixed dimensions. 
Both of these algorithms rely on packing/covering bounds that we prove in Section \ref{covering}. % , in contrast to building the quad tree in a preprocessing phase. 
% \aanote{A good transition here might be, that when we add in the additional $c_0$ constraint, we do infact get the global property
% that the ratio of any two interval lengths in the same level is bounded, and hence we can get the $\log(n)$ bound required. In practice, the former algorithm might still outperform the latter depending on your distribution and space. }

An important technical point is that for exposition and simplicity, we initially work with 
the \emph{symmetrized} Bregman divergences (of the form $\sbreg(x,y) = \breg(x \mid y) + \breg(y \mid x)$), 
and then  extend these results to general Bregman divergences (Section~\ref{sec:generalizations}).
 We note that the results for symmetrized Bregman divergences might be interesting in their own right, 
as they have also been used in applications~\cite{vptrees,tailoredbregmannn,symmetrizedcentroids,moresymmetrizedcentroids}.

%%% Local Variables: 
%%% mode: latex
%%% TeX-master: "paper"
%%% End: 

\section{Related Work}
\label{sec:related}

Approximate nearest-neighbor algorithms come in two flavors: the high dimensional variety, where all bounds must be polynomial in the dimension $d$, and the constant-dimensional variety, where terms exponential in the dimension are permitted, but query times must be sublinear in $n$. In this paper, we focus on the constant-dimensional setting. The idea of using ring-trees appears in many works ~\cite{indykmotwani,blackbox,peledmendel}, and a good exposition of the general method can be found in Har-Peled's textbook ~\cite[Chapter 11]{snotes}. 

The Bregman distances were first introduced by Bregman\cite{bregman}. They are the unique divergences that satisfy certain axiom systems for distance measures ~\cite{csiszar}, and are key players in the theory of information geometry ~\cite{amari}. Bregman distances are used extensively in machine learning, where they have been used to unify boosting with different loss functions\cite{collins2002logistic} and unify different mixture-model density estimation problems ~\cite{dhillon}. A first study of the algorithmic geometry of Bregman divergences was performed by Nielsen, Nock and Boissonnat ~\cite{bvd}. This was followed by a series of papers analyzing the behavior of clustering algorithms under Bregman divergences ~\cite{musimilarcoresets,ackermann2, ackermann3,roglin1,mcgregor}. 

Many heuristics have also been proposed for spaces endowed with Bregman divergences. Nielsen and Nock ~\cite{nsmallestdisk} developed a Frank-Wolfe-like iterative scheme for finding minimum enclosing balls under Bregman divergences. Cayton ~\cite{caytonpaper} proposed the first nearest-neighbor search strategy for Bregman divergences, based on a clever primal-dual branch and bound strategy. Zhang \etal ~\cite{bregsearch} developed another prune-and-search strategy that they argue is more scalable and uses operations better suited to use within a standard database system. 
For good broad reviews of near neighbor search in theory and practice, the reader is referred to the books by Har-Peled\cite{snotes}, Samet ~\cite{samet} and Shakhnarovich \etal ~\cite{nnbook}.

%%% Local Variables: 
%%% mode: latex
%%% TeX-master: "paper"
%%% End: 

\section{Definitions}
\label{sec:defn}

In this paper we study the approximate nearest neighbor problem for distance functions $D$:  Given a point set $P$, a query point $q$, and an error parameter $\eps$, find a point $\text{nn}_q \in P$ such that $D(\text{nn}_q,q) \leq (1 + \eps)\min_{p \in P} D(p,q)$. 
We start by defining general properties that we will require of our distance measures. In what follows, we will assume that the distance measure $D$ is \emph{reflexive}: $D(x,y) = 0$ if and only if $x = y$ and otherwise $D(x,y) > 0$. 

\begin{defn}[Monotonicity]\label{monotonedefn}
Let $M\subset \reals$, $D:M\times M\to\reals$ be a distance function. We say that $D$ is \emph{monotonic} if and only if for all $a<b<c$
we have that $D(a,b) \leq D(a,c)$ and $D(b,c) \leq D(a,c)$.
\end{defn}

For a general distance function $D : M \times M \to \reals$, where $M \subset \reals^d$, we say that $D$ is monotonic if it is monotonic when restricted to any subset of $M$ parallel to a coordinate axis. 

\begin{defn}[Reverse Triangle Inequality(RTI)]
Let $M$ be a subset of $\reals$.  We say that a monotone distance measure $D : M \times M \to \reals$ satisfies a \emph{reverse triangle inequality} or RTI if for any three elements $a \le b \le c \in M$, 
$ D(a,b) + D(b,c) \le D(a,c) $
\end{defn}

\begin{defn}[$\mu$-defectiveness] \label{musimdefn}
Let $D$ be a symmetric monotone distance measure satisfying the reverse triangle inequality. We say that 
$D$ is \emph{$\mu$-defective} with respect to domain $M$ if for all $a,b,q \in M$,
\begin{equation}
|D(a,q) - D(b,q)| <  \mu D(a,b)
\end{equation}

For an asymmetric distance measure $D$, we define left and right sided $\mu$-defectiveness respectively as 
\begin{equation}
|D(q,a) - D(q,b)| < \mu D(a,b)
\end{equation}

\begin{equation}|D(a,q) - D(b,q)| < \mu D(b,a)
\end{equation}

Note that by interchanging $a$ and $b$ and using the symmetry of the modulus sign, we can also rewrite left and right sided 
$\mu$-defectiveness respectively as  $|D(q,a) - D(q,b)| < \mu D(b,a)$ and $|D(a,q) - D(b,q)| < \mu D(a,b)$.
\end{defn}

\paragraph{Two technical notes.} The distance functions under consideration are typically defined over $\reals^d$. We will assume in this paper that the distance $D$ is \emph{decomposable}: roughly, that $D((x_1, \ldots, x_d), (y_1, \ldots, y_d))$ can be written as $g( \sum_i f(x_i, y_i))$, where $g$ and $f$ are monotone. This captures all the Bregman divergences that are typically used (with the exception of the Mahalanobis distance and matrix distances). See Table~\ref{tab:breg-examples}.

\ctable[
    caption = Commonly used Bregman divergences ,
    label = tab:breg-examples,
    pos = htbp
]{c|c|c|c}{
\tnote{The Mahalanobis distance is technically not decomposable, but is a linear transformation of a decomposable distance}
\tnote[b]{($S^d_{++}$ denotes the cone of positive definite matrices)}
}{
    Name & Domain & $\phi$ & \breg(x,y)\\ \hline
    $\ell_2^2$ & $\reals^d$ & $\frac{1}{2}\|x\|^2$  & $\frac{1}{2}\|x - y \|^2_2$ \\
    Mahalanobis\tmark & $\reals^d$ & $\frac{1}{2} x^\top Q x$& $\frac{1}{2}(x-y)^\top Q (x-y)$\\
    Kullback-Leibler & $\reals^d_+$& $\sum_i x_i \log x_i$& $\sum x_i \log \frac{x_i}{y_i} - x_i + y_i$\\
    Itakura-Saito & $\reals^d_+$& $-\sum_i \log x_i$& $\sum \Bigl( \frac{x_i}{y_i} - \log \frac{x_i}{y_i} - 1\Bigr)$\\
    Exponential & $\reals^d$& $\sum_i e^{x_i}$& $\sum e^{x_i} - (x_i - y_i +1)e^{y_i} $\\
    Bit entropy & $[0,1]^d$ & $\sum_i x_i \log x_i + (1-x_i) \log(1-x_i)$ & $\sum x_i \log \frac{x_i}{y_i} + (1-x_i) \log \frac{1-x_i}{1-y_i} $\\
    Log-det & $S^d_{++}\tmark[b]$& $\log \det X$& $\langle X, Y^{-1} \rangle - \log \det XY^{-1} - N $\\
    von Neumann entropy & $S^d_{++}$& $\text{tr} (X \log X - X)$& $ \text{tr}(X (\log X - \log Y) - X + Y)$\\ \hline
}

 We will also need to compute the diameter of an axis parallel box of side-length $\ell$. Our results hold as long as the diameter of such a box is $O(\ell d^{O(1)} )$: note that this captures standard distances like those induced by norms, as well as decomposable Bregman divergences. In what follows, we will mostly make use of the \emph{square root} of a Bregman divergence, for which the diameter of a box is $\ell (\mu+1) d^{\frac{1}{2}}$ or $\ell d^{\frac{1}{2}}$, and so without loss of generality we will use this in our bounds. 

\paragraph{Bregman Divergences.}

Let $\phi: M\subset \reals^d \to \reals$ be a \emph{strictly convex} function that is differentiable in the relative interior of $M$. Strict convexity implies that the second derivative is never $0$ and will be a convenient technical assumption.
The \emph{Bregman divergence} $\breg$ is defined as 
\begin{equation} \breg(x,y) = \phi(x) - \phi(y) - \langle \nabla \phi(y), x-y\rangle \end{equation}

In general, $\breg$ is asymmetric. A \emph{symmetrized} Bregman divergence can be defined by averaging: 
\begin{equation} \sbreg(x,y) = \frac{1}{2}(\breg(x,y) + \breg(y,x)) = \frac{1}{2}\langle x - y, \nabla \phi(x) - \nabla \phi(y) \rangle
\end{equation}

An important subclass of Bregman divergences are the \emph{decomposable} Bregman divergences. Suppose $\phi$ has domain $M = \prod_{i=1}^d M_i $ and can be written as $\phi(x) = \sum_{i=1}^d \phi_i(x_i)$, where $\phi_i :M_i \subset \reals \to \reals$ is 
also strictly convex and differentiable in relint($S_i$). Then $ \breg(x,y) = \sum_{i=1}^d D_{\phi_i}(x_i, y_i)$
is a \emph{decomposable} Bregman divergence. 

 The majority of commonly used Bregman divergences are decomposable:~\cite[Chapter 3]{cayton-thesis} illustrates some of the commonly used ones, including the Euclidean distance, the KL-divergence, and the Itakura-Saito distance.
 In this paper we will hence limit ourselves to considering decomposable distance measures. 
We note that due to the primal-dual relationship of  $\breg(a,b)$ and  $D_{\phi^*}(b^*, a^*)$, 
for our results on the asymmetric Bregman divergence we need only consider right-sided nearest neighbors and left-sided results follow symmetrically.

\paragraph{Some notes on terminology and computation model.}
We note now that whenever we refer to ``bisecting" an interval $[a b]$ under a distance measure $D$ satisfying an RTI, we shall precisely mean finding $x$ s.t $D(a,x) = D(x,b)$. The RTI now implies that $D(a,x) = D(x,b) \leq \frac{1}{2} D(a,b)$
and that repeated bisection quickly reduces the length of subintervals. Computing
such a bisecting point of an interval exactly, or even placing a point at a specified distance from a given point $p$ is not
trivial. However we argue in Section \ref{sec:numerical} that both tasks can be approximately done by numerical procedures
without significantly affecting our asymptotic bounds. For the remainder of the paper we shall take an idealized context and assume any such computations
can be done to the desired accuracy quickly. 

We also stipulate that the ``diameter" of any subset of our domain $X \subset M$ under distance measure $D$ shall be $\max_{x,y \in X} D(x,y)$.
Where the choice of distance measure $D$ may appear ambiguous from the context, we shall explicitly refer to the $D$-diameter.

\section{Properties of Bregman Divergences}\label{sec:properties}
\label{sec:prop-sqrtd_s-phi}

The previous section defined key properties that we desire of a distance function $D$. The Bregman divergences (or modifications thereof) 
satisfy the following properties, as can be shown by direct computation. 

\begin{lemma}\label{lefttr}
Any one-dimensional Bregman divergence is monotonic.
\end{lemma}
\begin{lemma}\label{cover}
Any one-dimensional Bregman divergence satisfies the reverse triangle inequality. 
Let $a \leq b \leq c$ be three points in the domain of $\breg$. Then it holds that:
\begin{equation}
\breg(a,b) + \breg(b,c) \leq \breg(a,c) 
\end{equation}

\begin{equation}
  \breg(c,b) + \breg(b,a) \leq \breg(c,a)
\end{equation}
\end{lemma}

\begin{proof}
We prove the first case, the second follows almost identically.

\begin{align*}
\breg(a,b) + \breg(b,c) &= \phi(a) - \phi(b) - \phi'(b)(a-b) + \phi(b) - \phi(c) - \phi'(c)(b-c) \\
&= \phi(a) - \phi(c) - \phi'(b)(a-b) - \phi'(c)(b-c)
\end{align*}

But since $\phi''(x) \geq 0$ for all $x \in \mathbb{R}$, by convexity of $\phi$ we have that $\phi'(b) \leq \phi'(c)$. This allows us
to make the substitution.

\begin{align*}
 \breg(a,b) + \breg(b,c) &= \phi(a) - \phi(c) - \phi'(b)(a-b) - \phi'(c)(b-c) \\
& \leq \phi(a) - \phi(c) - \phi'(c)(a-b) - \phi'(c)(b-c) \\
&= \phi(a) - \phi(c) - \phi'(c)(a-c) \\
&= \breg(a,c)
\end{align*}

\end{proof}

Note that this lemma can be extended similarly by induction to any series of $n$ points between $a$ and $c$.
 Further, using the relationship between $\breg(a,b)$ and the ``dual'' distance  $D_{\phi^*}(b^*, a^*)$, we can show that the reverse triangle inequality holds going ``left'' as well:  
$\breg(c,b) + \breg(b,a) \leq \breg(c,a)$. These two separate reverse triangle inequalities together yield the result for  $\sbreg$. We also
get a similar result for $\sqrt{\sbreg}$ by algebraic manipulations.

\begin{lemma}\label{Aklreverse}
$\sqrt{\sbreg}$ satisfies the reverse triangle inequality.
\end{lemma}
\begin{proof}

Fix $a \le x \le b$, and assume that the reverse triangle inequality does not hold:
\begin{align*}
\sqrt{\sbreg(a,x)} + \sqrt{\sbreg(x,b)} &> \sqrt{\sbreg(a,b)} 
\\ \sqrt{(x-a) (\phi'(x) - \phi'(a))} + \sqrt{(b-x) (\phi'(b) - \phi'(x))} &> \sqrt{(b-a) (\phi'(b) - \phi'(a))} 
\end{align*}

Squaring both sides, we get:
\begin{align*}
(x-a)( \phi'(x)-\phi'(a)) + (b-x) (\phi'(b) - \phi'(x)) \hspace{1.5in}&
\\ + 2 \sqrt{(x-a)(b-x)(\phi'(x) - \phi'(a))(\phi'(b) - \phi'(x))} &> (b-a)( \phi'(b) - \phi'(a)) 
\\ (b-x)(\phi'(x) - \phi'(a)) + (x-a)(\phi'(b) - \phi'(x)) \hspace{1.5in}&
\\ - 2 \sqrt{(x-a)(b-x)(\phi'(x) - \phi'(a))(\phi'(b) - \phi'(x))} &< 0 
\\ \left( \sqrt{(b-x)(\phi'(x) - \phi'(a))} - \sqrt{(x-a)(\phi'(b) - \phi'(x))} \right)^2 &< 0
\end{align*}
which is a contradiction, since the LHS is a perfect square.
\end{proof}

 While the Bregman divergences satisfy both monotonicity and the reverse triangle inequality,
 they are not $\mu$-defective with respect to \emph{any} domain! An easy example of this is $\ell_2^2$, which is also a Bregman divergence.
A surprising fact however is that $\sqrt{\sbreg}$ and $\sqrt{\breg}$ do satisfy $\mu$-defectiveness (with $\mu$ depending on the bounded size of our domain). While we were unable to show precise bounds for $\mu$ in terms of the domain, the values are small.  For example, for the symmetrized KL-divergence on the simplex where each coordinate is bounded between $0.1$ and $0.9$, $\mu$ is $1.22$. If each coordinate is between $0.01$ and $0.99$,then $\mu$ is $2.42$. We discuss the
empirical values of $\mu$ in greater detail in Appendix \ref{sec:muranges}. The proofs showing $\mu$ is bounded are somewhat tedious and not highly insightful, so we place those in the Appendix \ref{sec:bounded} for the interested reader.

\begin{lemma}\label{Arootmu}
Given any interval $I=[x_1 x_2]$ on the real line, there exists a finite $\mu$ such that  $\sqrt{\sbreg}$ is $\mu$-defective with respect to $I$. We require 
all order derivatives of $\phi$ to be defined and bounded over the closure of $I$, and $\phi''$ to be bounded away from zero.
\end{lemma}

\begin{proof}
Refer to \ref{app:1} in Appendix.
\end{proof}

We note that the result for $\sqrt{\breg}$ is proven by establishing the
 following relationship between $\breg(a,b)$ and $\breg(b,a)$ over a bounded interval $I \subset \reals$, and with some
further computation. 
\begin{lemma}\label{firstTosecond}
Given a Bregman divergence $\breg$ and a bounded interval $I \subset \reals$, $\sqrt{\breg(a,b)}/ \sqrt{\breg(b,a)}$ is 
bounded by a parameter $c_0$  $\forall a,b \in I$ where $c_0$ depends on the choice of divergence and interval. We also require 
the derivatives of $\phi$ to be defined and bounded over the closure of $I$, and $\phi''$ to be bounded away from zero.
\end{lemma}
\begin{proof}
By continuity, compactness and the strict convexity of $\phi$, we have that over a finite interval $I$ $ c_ 0 = \max_x \phi_i'' (x)/  \min_y \phi_i''(y) $ is bounded.
 Now by using the Lagrange form of $\sqrt{\breg(a,b)}$, we get that $\sqrt{\breg(a,b)}/ \sqrt{\breg(b,a)} < \sqrt{c_0}$
\end{proof}

\begin{lemma}\label{Arootmubreg}
Given any interval $I=[x_1 x_2]$ on the real line, there exists a finite $\mu$ such that  $\sqrt{\breg}$ is right-sided $\mu$-defective with respect to $I$.We require all order 
derivatives of $\phi$ to be defined and bounded over the closure of $I$, and $\phi''$ to be bounded away from zero. 
\end{lemma}

\begin{proof}
Refer to \ref{app:2} in Appendix.
\end{proof}

We extend our results to $d$ dimensions naturally now by showing that if $M$ is a domain such that $\sqrt{\sbreg}$ and $\sqrt{\breg}$ are $\mu$-defective with respect to the projection of $M$ onto each coordinate axis, then $\sqrt{\sbreg}$ and $\sqrt{\breg}$ are  $\mu$-defective with respect to all of $M$.  

\begin{lemma}\label{AallDmusim}
Consider three points, $a = (a_1 , \ldots, a_i, \ldots, a_d)$, $b = (b_1 , \ldots, b_i, \ldots, b_d)$, $q = (q_1 , \ldots, q_i, \ldots, q_d)$ such that $ | \sqrt{\sbreg(a_i, q_i)} - \sqrt{\sbreg(b_i,q_i)}| < \mu \sqrt{\sbreg(a_i, b_i)}, \forall 1 \leq i \leq d$. Then 

\begin{equation}
\left| \sqrt{\sbreg(a, q)} - \sqrt{\sbreg(b,q)} \right| < \mu \sqrt{\sbreg(a, b)}
\end{equation}

Similarly, if $| \sqrt{\breg(a_i, q_i)} - \sqrt{\breg(b_i,q_i)}| < \mu \sqrt{\breg(a_i, b_i)}, \forall 1 \leq i \leq d$. Then 

\begin{equation}
\left| \sqrt{\breg(a, q)} - \sqrt{\breg(b,q)} \right| < \mu \sqrt{\breg(b,a)}
\end{equation}

\end{lemma}

\begin{proof}
\begin{align*}
\left| \sqrt{\sbreg(a, q)} - \sqrt{\sbreg(b,q)} \right| &< \mu \sqrt{\sbreg(a, b)} 
\\ \sbreg(a,q) + \sbreg(b,q) - 2 \sqrt{\sbreg(a,q) \sbreg(b,q)}  &< \mu^2 \sbreg(a,b) 
\\ \sum_{i=1}^{d} \left(\sbreg(a_i,q_i) + \sbreg(b_i, q_i)\right) - 2 \sqrt{\sbreg(a,q) \sbreg(b,q)} &< \mu^2 \sum_{i=1}^{d} \sbreg(a_i,b_i) 
\\ \sum_{i=1}^{d} \left(\sbreg(a_i, q_i) + \sbreg(b_i,q_i) - \mu^2 \sbreg(a_i,b_i)\right) &< 2 \sqrt{\sbreg (a,q)\sbreg(b,q)}
\end{align*}
The last inequality is what we need to prove for $\mu$-defectiveness with respect to  $a,b,q$.
By assumption we already have $\mu$-defectiveness w.r.t each $a_i,b_i, q_i$, for every $1 \leq i \leq d$: 
\begin{align*}
 \sbreg(a_i,q_i) + \sbreg(b_i,q_i) - \mu^2 \sbreg(a_i,b_i) &< 2 \sqrt{\sbreg(a_i,q_i) \sbreg(b_i,q_i)} 
\\ \sum_{i=1}^{d} \left(  \sbreg(a_i,q_i) + \sbreg(b_i,q_i) - \mu^2 \sbreg(a_i,b_i) \right) 
&< 2 \sum_{i=1}^{d} \sqrt{\sbreg(a_i,q_i) \sbreg(b_i,q_i)} 
\end{align*}
So to complete our proof we need only show:
\begin{equation}\label{sec:prop-sqrts-skl}
\sum_{i=1}^{d} \sqrt{\sbreg(a_i,q_i)} \sqrt{\sbreg(b_i,q_i)} \leq \sqrt{\sbreg(a,q)}\sqrt{\sbreg(b,q)} 
\end{equation}
But notice the following:
\begin{align*}
\sqrt{\sbreg(a,q)} &= \left( \sum_{i=1}^{d}\sbreg(a_i,q_i) \right)^{\frac{1}{2}} 
                   = \left( \sum_{i=1}^{d} \left( \sqrt{\sbreg(a_i,q_i)} \right)^2 \right)^{\frac{1}{2}}
\\ \sqrt{\sbreg(b,q)} &= \left( \sum_{i=1}^{d}\sbreg(b_i,q_i) \right)^{\frac{1}{2}}
                      = \left( \sum_{i=1}^{d} \left( \sqrt{\sbreg(b_i,q_i)} \right)^2 \right)^{\frac{1}{2}}
\end{align*}
So inequality \ref{sec:prop-sqrts-skl} is simply a form of the Cauchy-Schwarz inequality, which states that for two vectors $u$ and $v$ in $\reals^d$, that $\left|\left< u, v \right>\right| \leq \|u\| \|v\|$, or that
\begin{equation*}
  \left| \sum_{i=1}^d u_iv_i \right| \leq 
    \left( \sum_{i=1}^du_i^2 \right)^{\frac{1}{2}} 
    \left( \sum_{i=1}^dv_i^2 \right)^{\frac{1}{2}}
\end{equation*}

The second part of the proposition can be derived by an essentially identical argument.
\end{proof}

\section{Packing and Covering Bounds}
\label{covering}
The aforementioned key properties (monotonicity, the reverse triangle inequality, decomposability, and $\mu$-defectiveness) can be 
used to prove packing and covering bounds for a distance measure $D$. We now present some of these bounds. 

\subsection{Covering bounds in $1$ dimension}
\begin{lemma}[Interval packing]\label{1dintersect}
Consider a monotone distance measure $D$ satisfying the reverse triangle inequality, an interval $[ab]$ such that $D(a,b) = s$ and a
 collection of disjoint intervals intersecting $[ab]$, where $I = \{[x x'] \mid [x x'] , D(x, x') \geq \ell\}$. Then  $|I| \leq \frac{s}{\ell}+2$.
\end{lemma}

\begin{proof}
 Let $I'$ be the intervals of $I$ that are totally contained in $[ab]$. 
 The combined length under $D$ of all intervals in $I'$ is at least $|I'|\ell$, but by the reverse triangle inequality their total length cannot exceed $s$, so $|I'| \leq \frac{s}{\ell}$. 
 There can be only two members of $I$ not in $I'$, so  $|I| \leq \frac{s}{\ell} + 2$.
 \end{proof}  

A simple greedy approach yields a constructive version of this lemma. 
\begin{corollary}\label{1dcover}
Given any two points, $a \leq b$ on the line s.t $D(a,b) = s$, we can construct a packing of $[ab]$ by $r \le \frac{1}{\eps}$ intervals $[x_i x_{i+1}]$, $1 \leq i \leq r$ such that $D(a, x_0) = D(x_i, x_{i+1}) = \eps s$, $\forall i$ and $D(x_r , b) \leq \eps s$. Here $D$ is a monotone distance measure satisfying the reverse triangle inequality.
\end{corollary}

We recall here that $\breg$, $\sbreg$ and $\sqrt{\sbreg}$ satisfy the conditions of Lemma \ref{1dintersect} and corollary \ref{1dcover} as they satisfy
an RTI and are decomposable.  However, since $\sqrt{\breg}$  may not satisfy the reverse triangle inequality, we instead prove a weaker 
packing bound on $\sqrt{\breg}$ by using $\breg$.

\begin{lemma}[Weak interval packing]\label{1dsqrtbregint}
 Given distance measure $\sqrt{\breg}$ and an interval $[ab]$ such that $\sqrt{\breg}(a,b) = s$ and a collection of disjoint intervals intersecting $[ab]$
where $I = \{[x x'] \mid [x x'] , \sqrt{\breg}(x, x') \geq \ell \}$. Then  $|I| \leq \frac{s^2}{\ell^2}+2$. Such a set of intervals can be explicitly constructed.
\end{lemma}

\begin{proof}
 We note that here $\breg(a,b)  = s^2$, and $I = \{[x x'] \mid [x x'] , \breg(x, x') \geq \ell^2 \}$. The result then follows trivially from lemma \ref{1dintersect},
since $\breg$ satisfies the conditions of lemma \ref{1dintersect}.
\end{proof}

\subsection{Properties of cubes and their coverings}
 The one dimensional bounds can be generalized to higher dimensions to provide packing bounds for 
balls and cubes (which we define below) with respect to a monotone, decomposable distance measure.

\begin{defn} \label{cube}
 Given a collection of $d$ intervals $a_i, b_i$ and distance measure $D$, s.t $D(a_i, b_i) = s$ where $1 \leq i \leq d$, the \emph{cube} in $d$ dimensions is defined as $\prod_{i=i}^{d} [a_i b_i]$ and is said to have \emph{side-length} $s$. We shall specify the choice of $D$ by referring to the cube as either a $\breg$-cube, $\sbreg$-cube, $\sqrt{\breg}$-cube, $\sqrt{\breg}$-cube or a
 $\sqrt{\sbreg}$-cube. Where we make an argument that holds for more than one of these types of cubes, we shall refer to simply a
 $D$-cube where the possible values of $D$ will be specified. We follow the same convention for balls.
 \end{defn}

  We add that for a given distance measure $D$, a \emph{box} $H$ can be defined similarly to
a cube, except that the side lengths need not necessarily be equal. In this case we let $H = \prod_{i=i}^{d} [a_i b_i]$
and let the $i$th side-length be $D(a_i,b_i)$. Again where the choice of distance measure $D$ appears at all ambiguous we shall refer to the $D$ side-length. 

We pause here to note that for an asymmetric decomposable measure $D$ in $d$ dimensions, every $D$-box has an implied associated ordering on each of the $d$ composing intervals. For a $D$-box defined as prod $\prod_{i=1}^d [a_i b_i]$ and bisected by a collection of $x_i$
such that $D(a_i,x_i) = D(x_i,b_i)$, there will be $2^d$ subboxes produced such that their $i$th composing interval will
be either $[a_i x_i]$ or $[x_i b_i]$.

\begin{figure}[H]
  \begin{center}
    \includegraphics[scale = 0.7]{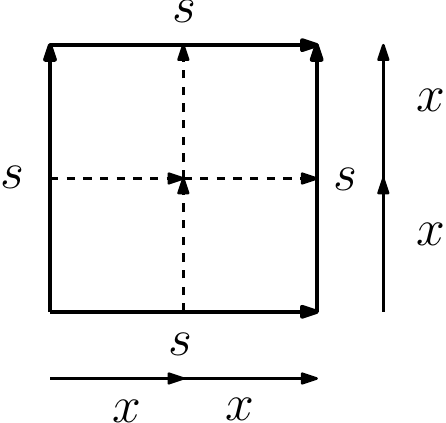}
  \end{center}
  \caption{A cube of directed side length $s$ subdivided into cubes of side length $x \leq \frac{s}{2}$ }
  \label{fig:bisection}
\end{figure}

 In what can be viewed as a generalization of bisection to splitting each side of a $D$-cube into multiple sub-intervals,
 we show the following:
 
\begin{lemma}\label{scubeCover}
Given a $d$ dimensional $D$-cube $B$ of side-length $s$ under distance measure $D$, 
we can cover it with at most $\eps^d$ $D$-cubes of side-length \emph{exactly} $\eps s$ under the same measure $D$, where $D$ may be either $\breg$, $\sbreg$ and $\sqrt{\sbreg}$. 
\end{lemma}

\begin{proof}
Note that $\sbreg$, $\breg$ and $\sqrt{\sbreg}$ satisfy conditions of corollary \ref{1dcover}. 
Hence we can employ the packing of at most $\frac{1}{\eps}$ points in each dimension spaced $\eps s$ apart. 
We then take a product over all $d$ dimensions, and the lemma follows trivially. 
\end{proof}
 
 Weaker packing bounds for $\sqrt{\breg}$ as noted in lemma \ref{1dsqrtbregint} yield us a weaker version of lemma \ref{scubeCover}.
\begin{lemma}\label{acubeCover}
Given a $d$ dimensional $\sqrt{\breg}$-cube $B$ of side-length $s$, 
we can cover it with at most $\eps^{2d}$ $\sqrt{\breg}$-cubes of side-length \emph{exactly} $\eps s$.
\end{lemma}
\begin{proof}
Identical to the proof of lemma \ref{scubeCover} and using lemma \ref{1dsqrtbregint} to obtain packing bounds.
\end{proof}
We note that this subdivision of $D$-cubes corresponds to placing an equal number of points (the vertices of the cubes),
and this is what we shall refer to more loosely as \emph{gridding} in the remainder of our paper.

\subsection{Covering with balls in higher dimensions}
Covering a $D$-ball with a number of smaller $D$-balls is a key ingredient in our results. Our approach is to divide a $D$-ball into $2^d$
orthants, then to show each orthant can be covered by a certain number of smaller $D$-cubes, and then finally that each
such $D$-cube
can be covered by a $D$-ball of a certain radius. 

We show now results for $\sbreg$,
$\breg$, $\sqrt{\breg}$ and $\sqrt{\sbreg}$. We present first the easier cases for the two symmetric measures, $\sbreg$ and
$\sqrt{\sbreg}$. 

\begin{lemma}\label{scube}
A $\sbreg$-cube in $d$ dimensions of side-length $s$ can be covered by a $\sbreg$-ball of radius $ds$. 
Similarly, a $\sqrt{\sbreg}$-cube in $d$ dimensions of side-length $s$ can be covered by a $\sqrt{\sbreg}$-ball
of radius $\sqrt{d} s$. 
\end{lemma}

\begin{proof}
Recall that a $\sbreg$-cube is defined as $\prod_{i=1}^d [a_i b_i]$ s.t $D_{s \phi_i}(a_i,b_i) = s$ (where $D_{s \phi_i}(a_i,b_i)$ is
induced by restricting $\sbreg$ to the $i$th dimension). Let the vertex space of the $\sbreg$-cube be $V = \prod_{i=1}^d v_i$, where $v_i \in \{a_i, b_i \}$. Now pick an arbitrary vertex $x \in V$, and consider
the $\sbreg$-ball $B$ of radius $ds$ with center $v$. By decomposability and monotonicity, for any $y \in V$, we have:
\begin{align*}
\sbreg(x,y) &= \sum_{i=1}^d D_{s \phi_i} (x_i,y_i) \leq \sum_{i=1}^d D_{s \phi_i} (a_i,b_i)\\
			 &= \sum_{i=1}^d s = ds			
\end{align*} 
Hence an $\sbreg$-cube of side-length $s$ can be covered by an $\sbreg$-ball of radius $ds$.  The second result follows
by noting that an $\sqrt{\sbreg}$-cube of side-length $s$ is an $\sbreg$-cube of side-length $s^2$. Hence this can be covered
by an $\sbreg$-ball of radius $d s^2$, which is simply an $\sqrt{\sbreg}$ ball of radius $\sqrt{d} s$.
\end{proof}

\begin{lemma}\label{acube}
A $\breg$-cube in $d$ dimensions of side-length $s$ can be covered by a $\breg$-ball of radius $ds$. 
Similarly, a $\sqrt{\breg}$-cube in $d$ dimensions of side-length $s$ can be covered by a $\sqrt{\sbreg}$-ball
of radius $\sqrt{d} s$. 
\end{lemma}

\begin{proof}
Similar to lemma \ref{scube}, we begin by recalling that a $\breg$-cube is defined as $\prod_{i=1}^d [a_i b_i]$ s.t $D_{ \phi_i}(a_i,b_i) = s$ (where $D_{ \phi_i}(a_i,b_i)$ is
induced by restricting $\breg$ to the $i$th dimension). We again let the vertex space of the $\breg$-cube be $V = \prod_{i=1}^d v_i$, where $v_i \in \{a_i, b_i \}$. Now we have to be somewhat more careful in our choice of center for the $\breg$-ball $B$ of
radius $ds$ than we were in lemma \ref{scube}. We choose the ``lowest" point of the $\breg$-cube, which is $x = \prod_{i=1}^d a_i$ (see figure \ref{fig:acube})
and term this as a \emph{canonical corner}.We note that our definition does not require that $a_i \leq b_i$. Now for any other $y \in V$ we have:
\begin{align*}
\breg(x,y) &= \sum_{i=1}^d D_{ \phi_i} (x_i,y_i) \leq \sum_{i=1}^d D_{ \phi_i} (a_i,b_i)\\
			 &= \sum_{i=1}^d s = ds			
\end{align*} 
The argument for $\sqrt{\breg}$ follows analogously to that for $\sqrt{\sbreg}$ in lemma \ref{scube}. 

\begin{figure}[H]
  \begin{center}
    \includegraphics[scale = 0.4]{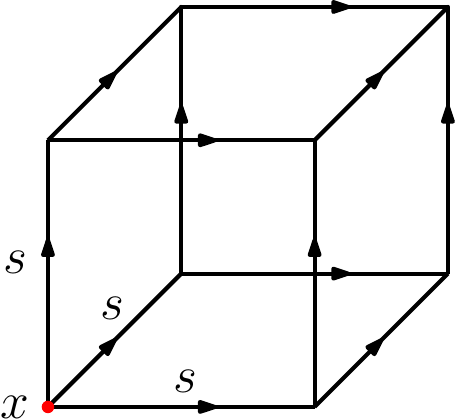}
  \end{center}
  \caption{$x$ is within $ds$ distance under $\breg$ of every other point of the cube.}
  \label{fig:acube}
\end{figure}
\end{proof}

We will also find the following relation of the diameter of a $\sqrt{\sbreg}$-cubes to the $\sqrt{\breg}$ side-length useful later in this paper.

\begin{lemma}\label{diamtolen}
The diameter of an $\sqrt{\sbreg}$-cube of side-length $s$ is bounded by $\sqrt{d} s$. 
\end{lemma}
\begin{proof}
Consider any two points $x$ and $y$ in the $\sqrt{\sbreg}$-cube of $\sqrt{\sbreg}$-side-length $s$ 
and defined as $\prod_{i=1}^d [a_i b_i]$. Note that since $x_i, y_i \in [a_i b_i]$ we have that $\sbreg(x_i,y_i) \leq s^2$.
Hence $\sbreg(x,y) \leq ds^2 $ and $\sqrt{\sbreg}(x,y) \leq \sqrt{d} s$.
\end{proof}
\begin{corollary}
For any $\sqrt{\sbreg}$-box of maximum $\sqrt{\sbreg}$-side length $s$, the diameter of the box is bounded by $\sqrt{d} s$.
\end{corollary}

We now proceed to showing covering bounds for $\sqrt{\sbreg}$ and $\sqrt{\breg}$ using the geometry we have developed thus far.

\begin{lemma}\label{sballcover}
Consider a $D$-ball $B$ of radius $s$ and center $c$.
Then in the case of $D=\sbreg$, $B$ can be covered with $\frac{2^d}{\eps^d}$ $\sbreg$-balls of radius $d \eps s$. In the case of $D=\sqrt{\sbreg}$, $B$ can be covered with 
$\frac{2^d}{\eps^d}$ $\sqrt{\sbreg}$-balls of radius $\sqrt{d} \eps s$. 
\end{lemma}
\begin{proof}
We divide the $D$-ball into $2^d$ orthants around the center $c$. Each orthant can be covered by a $D$-cube of size $s$. 
For both $D=\sbreg$ and $D=\sqrt{\sbreg}$, by lemma \ref{scubeCover} each such $D$-cube can be broken down 
into $\frac{1}{\eps^d}$ sub $D$-cubes of side-length $\eps s$.

By lemma \ref{scube}, we can cover each such $\sbreg$-cube by a $\sbreg$-ball of radius $d \eps s$ placed at any corner. Similarly, for $\sqrt{\sbreg}$, we can cover each
sub $\sqrt{\sbreg}$-cube by a $\sqrt{\sbreg}$-ball of radius $\sqrt{d} \eps {s}$ placed at any corner.

Since there are $\frac{1}{\eps^d}$ sub $D$-cubes to each of the $2^d$ orthants whether $D=\sqrt{\sbreg}$ or
$D=\sbreg$ respectively, the lemma now follows by covering each sub $D$-cube with a $D$-ball of the required radius.
\end{proof}

\begin{lemma}\label{aballcover}
Consider a $D$-ball $B$ of radius $s$ and center $c$ with respect to distance measure $D$.
Then in the case of $D= \breg$, $B$ can be covered with $\frac{2^d}{\eps^d}$ $\breg$-balls of radius $d \eps s$. And for  $D=\sqrt{\breg}$, $B$ can be 
covered by $\frac{2^d}{\eps^{2d}}$ $\sqrt{\breg}$-balls of radius $\sqrt{d} \eps s$. 
\end{lemma}

\begin{proof}
We divide the $D$-ball into $2^d$ orthants around the center $c$. Each orthant can be covered by a $D$-cube of size $s$. We now consider each case separately.
For $\breg$,  by lemma \ref{scubeCover} each such $\breg$-cube can be broken down 
into $\frac{1}{\eps^d}$ $\breg$-cubes of side-length $\eps s$. For $\sqrt{\breg}$, by lemma \ref{acubeCover} we can break down each $\sqrt{\breg}$-cube into 
$\frac{1}{\eps^{2d}}$ sub $\sqrt{\breg}$-cubes of side-length $\eps s$.

By lemma \ref{acube}, we can cover each such $\breg$-cube by a $\breg$-ball of radius $d \eps s$ placed at a canonical corner. Similarly for $\sqrt{\breg}$, by lemma \ref{acube} we can cover each
sub $\sqrt{\breg}$-cube by a $\sqrt{\breg}$-ball of radius $\sqrt{d} \eps {s}$ placed at a canonical corner. 
Since there are $\frac{1}{\eps^{2d}}$ and $\frac{1}{\eps^d}$ sub $D$-cubes to each of the $2^d$ orthants for $D=\sqrt{\breg}$ and
$D=\breg$ respectively, the lemma now follows by covering each sub $D$-cube with a $D$-ball of the required radius.
\end{proof}
\section{Computing a rough approximation}\label{sec:ringsec}
To illustrate our techniques, we will focus on finding approximate nearest neighbors under $\sqrt{\sbreg}$ over the next
two sections. When we define our notation more generally - e.g, of a ring separator - we may use a more generic distance measure
$D$.

Later we will show how our results can be extended to the asymmetric case with mild modifications and careful attention to
directionality.
We now describe how to compute a $O(\log n)$ 
rough approximate nearest-neighbor under $\sqrt{\sbreg}$ on our point set $P$, 
which we will use in the next section to find the $(1+\eps)$-approximate nearest neighbor. 
The technique we use is based on ring separators. Ring separators are a fairly old concept in geometry, 
notable appearances of which include the landmark paper by Indyk and Motwani~\cite{indykmotwani}. 
Our approach here is heavily influenced by Har-Peled and Mendel~\cite{peledmendel}, 
and by Krauthgamer and Lee~\cite{blackbox}, 
and our presentation is along the template of the textbook by Har-Peled~\cite[Chapter 11]{snotes}.
 
We note here that the constant of $d^{d/2}$ which appears in our final bounds for storage and query time
 is specific to $\sqrt{\sbreg}$. However, an argument on the same lines will yield a constant of $d^{O(d)}$
for any generic $\mu$-defective, symmetric RTI-satisfying decomposable distance measure $D$ such that 
the $D$-diameter of a cube of side-length $1$ is bounded by $d^{O(1)}$. 

Let $B(m,r)$ denote a $D$-ball of radius $r$ centered at $m$, and let $B'(m,r)$ denote the complement
 (or exterior) of $B(m,r)$. A \emph{ring} $R$ is the difference of two concentric $D$-balls: 
$R = B(m, r_2) \setminus B(m, r_1), r_2 \ge r_1$. We will often refer to the larger $D$-ball $B(m, r_2)$ 
as $B_{\text{out}}$ and the smaller $D$-ball as $B_{\text{in}}$. We use $P_{\text{out}}(R)$ to denote the set
 $P \cap B_{\text{out}}'$, and use $P_{\text{in}}(R)$ as $P \cap B_{\text{in}}$, where
 we may drop the reference to $R$ when the context is obvious. 
A \emph{$t$-ring separator} $R_{P,c}$ on a point set $P$ is a ring such that 
$\frac{n}{c} < |P_{\text{in}}| < (1 - \frac{1}{c})n $, $\frac{n}{c} < |P_{\text{out}}| < (1 - \frac{1}{c})n$, 
$r_2 \geq (1 + t) r_1$ and $B_{\text{out}} \setminus B_{\text{in}}$ is empty of points of $P$ (see figure \ref{eRing}). A $t$-ring tree is a binary tree obtained
by repeated dispartition of our point set $P$ using a $t$-ring separator. (We shall make the choice of distance measure $D$
explicit whenever using a $t$-ring separator.)
\begin{figure}[H]
  \begin{center}
    \includegraphics[scale = 0.4]{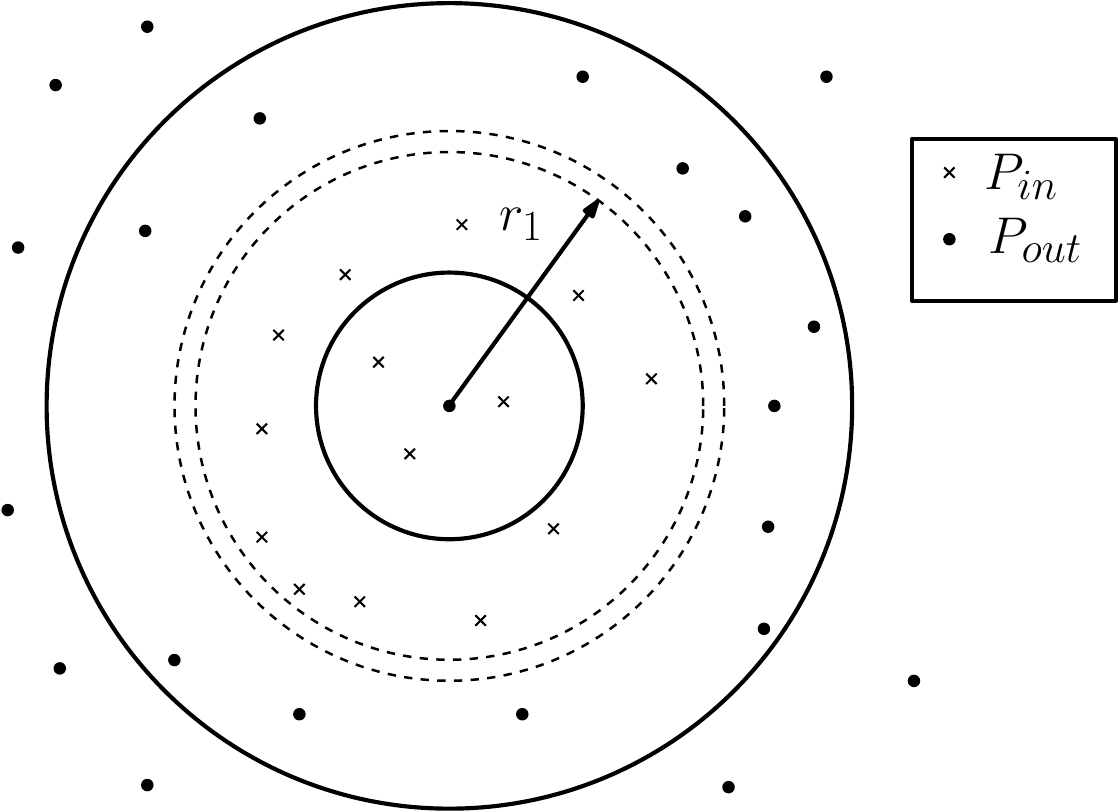}
  \end{center}
  \caption{The points of $P$ are split into $P_{\text{in}}$ and $P_{\text{out}}$}
  \label{eRing}
\end{figure}

Note that later on in this section, we will abuse this notation slightly by using ring-separators where the annulus
is not actually empty, but we will bound the added space complexity and tree depth introduced. Finally,  denote the minimum sized $D$-ball containing at least $\frac{n}{c}$ points of $P$ by $B_{\text{opt},c}$; 
its radius is denoted by $r_{\text{opt},c}$. 

We demonstrate that for any point set $P$ a ring separator exists under $\sqrt{\sbreg}$ and secondly, it can always be computed efficiently.  
Applying this ``separator'' recursively on our point structure yields a ring-tree structure for searching our point set.  Before we proceed further, we need to establish some properties of disks under a $\mu$-defective distance.
Lemma \ref{circle} is immediate from the definition of $\mu$-defectiveness, Lemma \ref{randommufraction}
is similar to one obtained by Har-Peled and Mazumdar~\cite{smallestdisk} and the idea of 
repeating points in both children of a ring-separator derives from a result by Har-Peled and 
Mendel~\cite{peledmendel}. 

\begin{lemma} \label{circle}
Let $D$ be a $\mu$-defective distance, and let $B(m,r)$ be a $D$-ball. Then for any two points $x,y \in B(m,r)$, $D(x,y) < (\mu+1) r$.
\end{lemma}
\begin{proof}
 Follows from the definition of  $\mu$-defectiveness.
  \begin{align*}
  D(x,y) - D(m,y) &< \mu  D(m,x) 
  \\ D(x,y) <  \mu r +  D(m,y)  & \leq (\mu + 1) r.
  \end{align*}
 \end{proof}
 
 \begin{corollary}
 For any $\sqrt{\sbreg}$-ball $B(m,r)$ and two points $x,y \in B(m,r)$,
 $\sqrt{\sbreg}(x,y) < (\mu+1) r$. 
 \end{corollary}
\begin{proof}
Since $\sqrt{\sbreg}$ is $\mu$-defective over a prespecified restricted domain.
\end{proof}

\begin{lemma} \label{randommufraction}
 Given a parameter $1 \leq c \leq n$, we can compute in $O(nc)$ expected time a $\mu+1$ approximation to the smallest  
radius $\sqrt{\sbreg}$-ball containing $\frac{n}{c}$ points by the algorithm \ref{smallest}. 
\end{lemma}
\begin{proof}
As described by Har-Peled and Mazumdar (\cite{smallestdisk}) we let $S$ be a random sample from $P$, generated by choosing every point of $P$ with probability $\frac{c}{n}$. Next, compute for every $p \in S$, the
smallest $\sqrt{\sbreg}$-ball centered at $p$ containing $c$ points of $P$. By median selection, this can be done in $O(n)$ time and since $E(|S|) = c$, this gives us
the expected running time of $O(nc)$. Now, let $r'$ be the minimum radius computed.
Note that by lemma \ref{circle}, if $|S \cap B_{\text{opt},c} | > 0$ then we have that $r' \leq (\mu+1) r_{opt}$.
 But since $B_{\text{opt},c}$ contains $\frac{n}{c}$ points, we can upper bound the probability of
 failure as the probability that we do not select any of the $\frac{n}{c}$ points in $B_{\text{opt}}$ in our
sample. Hence:

\begin{align*}
Pr(|S \cap B_{\text{opt},c}| > 0) = 1 - \left(1 - \frac{c}{n} \right)^{\frac{n}{c}} \geq 1 - \frac{1}{e}
\end{align*}

Note that one can obtain a similar approximation deterministically by brute force search, but this would incur a prohibitive
$O(n^2)$ running time.
\end{proof}

\begin{algorithm}
  \caption{ApproxSmallestBall$(P,c)$}
  \begin{algorithmic}
  	\STATE  $n \gets |P|$
    \STATE Choose $S$ by picking every $p \in P$ with probability $\frac{n}{c}$
    \STATE  $r \gets \infty$
    \STATE $B \gets \text{NULL}$
    \FORALL {$s \in S$}     	
    	\STATE Compute smallest $\sqrt{\sbreg}$-ball $B(s,r_s)$ with center $s$ that contains $c$ points of $P$.
    	\IF{ $r_s < r$}
    	  \STATE $B \gets B(s,r_s)$
    	  \STATE $r \gets r_s$
    	\ENDIF 
    \ENDFOR 
    \RETURN $B$
  \end{algorithmic}
  \label{smallest}
\end{algorithm}

We can now use Lemma \ref{randommufraction} and the corresponding algorithm \ref{construct} to construct our ring-separator.

\begin{algorithm}
  \caption{MakeRing$(P,t)$}
  \begin{algorithmic}
  \STATE \COMMENT{Here $t>1$, and the thickness of the separating ring is $O \left(\frac{1}{t} r_{\text{in}} \right)$}
  	\STATE  $n \gets |P|$
  	\STATE $\text{NODE IN} \gets \text{NULL}$
  	\STATE $\text{NODE OUT} \gets \text{NULL}$
    \STATE $c \gets 2(4 (\mu+1) \sqrt{d})^d$
    \STATE  $B_1(m_1, r_1) \gets \text{ApproxSmallestBall}(P,c)$
    \STATE  $B_2(m_2,r_2) \gets B(m_1 , 2 r_1)$
    \STATE $\text{ANNULUS} \gets B_2 \setminus B_1$
    \STATE Divide $\text{ANNULUS}$ into $t$ rings of equal thickness, such that $\text{RINGS}[i]$ is the $i$-th ring.
    \STATE $\text{COUNT} \gets \infty$
    \STATE $r_{\text{in}} \gets r_1$
    \FORALL {$i =1 \to t$}
    	\IF{  $|P \cap \text{RINGS}[i]| < \text{COUNT}$}
    		\STATE $\text{COUNT} \gets |P \cap \text{RINGS}[i]|$
    		\STATE $r_{\text{in}} \gets r_1 + \left(\frac{i-1}{t}\right)r_1$
    	\ENDIF
    \ENDFOR
    \FORALL {$p \in P$}
    	\IF{$\sqrt{\sbreg}(m_1,p) \leq r_{\text{in}}$}
    		\STATE Add $p$ to $\text{IN}$
    	\ELSIF{$\sqrt{\sbreg}(m_1,p) \geq r_{\text{in}}+ \frac{r_1}{t}$}
    		\STATE Add $p$ to $\text{OUT}$
    	\ELSE
    		\STATE Add $p$ to $\text{IN}$ and $\text{OUT}$
    	\ENDIF
    \ENDFOR
    \STATE $\text{COUNT-IN} \gets$ number of points in $\text{IN}$
    \STATE $\text{COUNT-OUT} \gets$ number of points in $\text{OUT}$
    \IF {{$\text{COUNT-IN} \geq \left(1 - \frac{1}{c} \right)n$} \OR {$\text{COUNT-OUT} \geq \left(1 -\frac{1}{c} \right)n$}} 
    	\RETURN {MakeRing $(P,t)$} 
    	
    \COMMENT{This checks implictly that our earlier call to the randomized $\text{ApproxSmallestBall}(P,c)$
    returned our desired approximation. If not, we try our procedure again as standard for Las Vegas algorithms.}
    \ELSE	
    	\RETURN {$\text{IN}$ and $\text{OUT}$}
    \ENDIF
  \end{algorithmic}
  \label{construct}
\end{algorithm}

\begin{lemma} \label{improvedRing}
For arbitrary $t$ s.t $1 < t < n$ and $\sqrt{\sbreg}$ in a $\mu$-defective domain, we can construct a $\frac{1}{t}$-ring separator $R_{P,c}$ under $\sqrt{\sbreg}$ in $O(n)$ expected time on a point set $P$ by repeating points. See algorithm \ref{construct}.
\end{lemma}
\begin{proof}

Using Lemma \ref{randommufraction}, we compute a $\sqrt{\sbreg}$-ball $S=B(m,r_1)$ (where $m \in P$) 
containing $\frac{n}{c}$ points such that $r_1 \leq (\mu+1)r_{\text{opt},c}$ where $c$ is a parameter to be set. 
Consider the $\sqrt{\sbreg}$-ball $\bar{S} = B(m,2r_1)$. We shall argue that there 
must be $\frac{n}{c}$ points of $P$ in the complement of $\bar{S}$, $\bar{S}'$, for careful choices of $c$.
As described in Lemma \ref{sballcover}, $\bar{S}$ can be covered by $2^d$ hypercubes of side-length 
$2 r_1$, the union of which we shall refer to as $H$. Set $L = (\mu+1)\sqrt{d}$.  Imagine a partition of $H$ 
into a grid, where each cell is of $\sqrt{\sbreg}$-side-length $\frac{r_1}{L}$ and hence of diameter at 
most $\Delta(\frac{r_1}{L}, d) = \frac{r_1}{\mu+1} \leq r_{\text{opt},c}$ (by lemma \ref{diamtolen}). 
A $\sqrt{\sbreg}$-ball of radius $r_{\text{opt},c}$ on any corner of a cell will contain the entire cell, and so it 
will contain at most $\frac{n}{c}$ points, by the definition of $r_{\text{opt},c}$. 

 By Lemma \ref{scubeCover} the grid on $H$ has at most $2^d(2r_1/\frac{r_1}{L})^d = (4 (\mu+1) \sqrt{d})^d$ cells.
 Set $c = 2(4 (\mu+1) \sqrt{d})^d$.  Then we have that $\bar{S} \subset H$ 
contains at most $\frac{n}{c} (4 (\mu+1) \sqrt{d})^d  = \frac{n}{2}$ points.  
Since the inner $\sqrt{\sbreg}$-ball $S$ contains at least $\frac{n}{c}$ points, and the outer $\sqrt{\sbreg}$-ball $\bar{S}$ 
contains at most $\frac{n}{2}$ points, hence the annulus $\bar{S} \setminus S$ contains at most 
$\frac{n}{2} - \frac{n}{c}$ points. Now, divide $\bar{S} \setminus S$ into $t$ rings of equal width, 
and by the pigeonhole principle at least one of these rings must contain at most $ O(\frac{n}{t})$ 
points of $P$. Now let the inner $\sqrt{\sbreg}$-ball corresponding to this ring be $B_{\text{in}}$ and the outer 
$\sqrt{\sbreg}$-ball 
be $B_{\text{out}}$. Let $P_{\text{in}} = P \cap B_{\text{in}}$, $P_{\text{out}} = P \cap B_{\text{out}}'$. Add any remaining points of $P$ to \emph{both} $P_{\text{in}}$ and $P_{\text{out}}$(see figure \ref{aRing}),i.e,
consider that these points are duplicated and are in both sets.

\begin{figure}[H]
  \begin{center}
    \includegraphics[scale = 0.6]{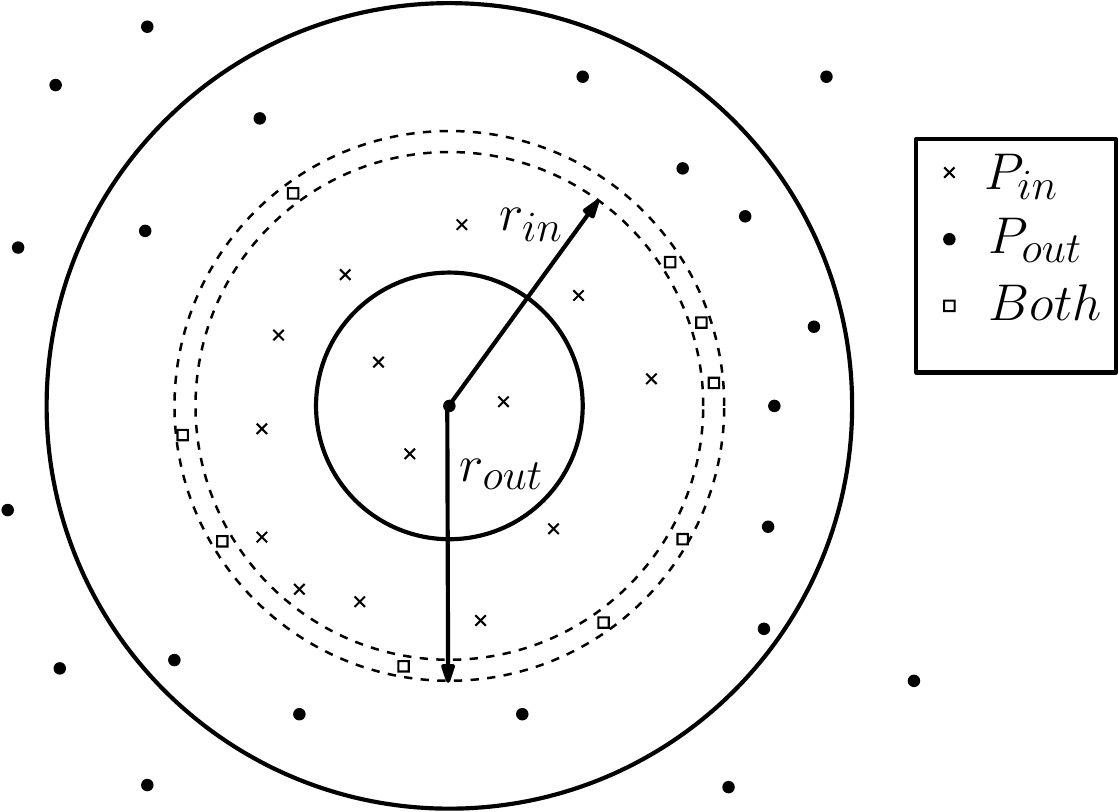}
  \end{center}
  \caption{The points are split into $P_{\text{in}}$ and $P_{\text{out}}$ with some point duplication}
  \label{aRing}
\end{figure}

Assign $P_{\text{in}}$ and $P_{\text{out}}$ to two nodes $v_{\text{in}}$ and $v_{\text{out}}$ respectively.  Even for $t = 1$,  each node contains at most 
$\frac{n}{2} + (\frac{n}{2} - \frac{n}{c}) = (1 - \frac{1}{c})n$ points. Also, the thickness of the ring is bounded by $\frac{2 r_1 - r_1}{t}/2r_1 = \frac{1}{2t}$, i.e it 
is a $O(\frac{1}{t})$ ring separator.
Finally, we can check in $O(n)$ time if the 
randomized process of Lemma \ref{randommufraction} succeeded simply by verifying the number of points
in the inner and outer ring is bounded by the values just computed.
\end{proof}

\begin{lemma}\label{ringsep}
Given any point set $P$ under $\sqrt{\sbreg}$ in a $\mu$-defective domain, we can construct a $O(\frac{1}{\log n})$ ring-separator tree $T$ of depth $O(d^{\frac{d}{2}} (\mu+1)^d \log n)$ by algorithm \ref{makeTree}.
\end{lemma}

\begin{proof}
Repeatedly partition $P$ by lemma \ref{randommufraction} into $P^{v}_{\text{in}}$ and $P^{v}_{\text{out}}$ where $\textbf{v}$ is the parent node. Store only the single point $\text{rep}_v = m \in P$ in node $\textbf{v}$, the center of the $\sqrt{\sbreg}$-ball $B(m,r_1)$. 
We continue this partitioning until we have nodes with only a single point contained in them.
Since each child contains at least $\frac{n}{c}$ points (by proof of Lemma \ref{improvedRing}), each subset reduces by a factor of at least $1 - \frac{1}{c}$ at each step, 
and hence the depth of the tree is logarithmic. We calculate the depth more exactly, 
noting that in Lemma \ref{improvedRing}, $c = O(d^{\frac{d}{2}} (\mu+1)^d)$.
Hence the depth $x$ can be bounded as:
\begin{align*}
n (1 - \frac{1}{c})^x &= 1 
\\ (1- \frac{1}{c})^x &= \frac{1}{n} 
\\ x &
= \frac{ \ln \frac{1}{n}}{\ln (1 - \frac{1}{c})} 
= \frac{-1}{\ln (1 - \frac{1}{c})} \ln n
\\ x &\leq c \ln n 
 = O \left( d^{\frac{d}{2}} (\mu+1)^d \log n \right)
\end{align*}
\end{proof}

\begin{algorithm}
  \caption{MakeTree$($P$, \text{NODE ROOT}, t)$}
  \begin{algorithmic}
  \STATE \COMMENT{Here $t <1$ is the thickness of the ring w.r.t radius of the inner ball.}
  	\STATE Add $P$ to $\text{ROOT}$
    \STATE $(\text{IN}$, $\text{OUT})$ $\gets$ MakeRing $\left(P, \frac{1}{t} \right)$
	\STATE Set $\text{IN}$ as a child of $\text{ROOT}$
	\STATE Set $\text{OUT}$ as a child of $\text{ROOT}$
	\STATE MakeTree($P \cap \text{IN}$ , $\text{IN}$, $t$)
	\STATE MakeTree($P \cap \text{ OUT}$ , $\text{OUT}$, $t$)
  \end{algorithmic}
  \label{makeTree}
\end{algorithm}

Finally, we verify that the storage space require is not excessive.
\begin{lemma}\label{ringstorage}
To construct a $O(\frac{1}{\log n})$ ring-separator tree under $\sqrt{\sbreg}$ in a $\mu$-defective domain requires $O(n)$ storage and $O(d^{\frac{d}{2}} (\mu+1)^d n \log n)$ time.
\end{lemma}

\begin{proof}
By Lemma \ref{ringsep} the depth bounds still hold upon repeating points.  For storage, we have to bound
the total number of points in our data structure after repetition, let us say $P_R$. 
Since each node corresponds to a splitting of $P_R$,there may be only $O(P_R)$ nodes and total storage. We aim to show $|P_R| = O(|P|) = O(n)$. We begin by noting that in the proof of Lemma \ref{improvedRing}, for a node containing $x$ points, at most an additional $\frac{x}{\log n}$ may be duplicated in the two children.

To bound this over each level of our tree, we sum across each node to obtain that the number of points $T_i$ in our structure at the $i$-th level, as: 
\begin{equation}
T_i = T_{i-1} \left( 1 + \frac{1}{\log T_{i-1}} \right)
\end{equation} 
Note also by Lemma \ref{ringsep}, the tree depth is $O(\log n)$ 
or bounded by $k \log n$ where $k$ is a constant. Hence we only need to bound the storage at the level $i = O(\log n)$. We solve the recurrence, noting that $T_0 = |P| = n$ (no points have been duplicated yet) and $T_i > n$ for all $i$ and hence 
$T_i < T_{i-1}(1 + \frac{1}{\log n})$. Thus the recurrence works out to:
\begin{align*}
T_i &< n \left( 1+ \frac{1}{\log n} \right)^{O(\log n)} 
< n \left( \left( 1+ \frac{1}{\log n} \right)^{\log n} \right)^k 
< n(e^k).
\end{align*}

Where the main algebraic step is that $(1+ \frac{1}{x})^x < e$. 
This proves that the number of points, and hence our storage complexity is $O(n)$. Multiplying the depth by $O(n)$ for computing the smallest under $\sqrt{\sbreg}$-ball across nodes on each level, gives us the time complexity of $O(n \log n)$.
We note that other tradeoffs are available for different values of approximation quality ($t$) and construction time / query time.
\end{proof}

\paragraph{Algorithm and Quality Analysis}
\begin{algorithm}
  \caption{FindRoughNN$($P$, q, t, \text{NODE ROOT}$}
  \begin{algorithmic}
  	\STATE $\text{curr} \gets \text{rep}(\text{ROOT})$
  	\STATE $\text{best}_q \gets \text{curr}$
    \STATE $D_{\text{near}} \gets \sqrt{\sbreg}(q,\text{curr})$
	
	\WHILE {$\text{ROOT}$ has children}
	\STATE $\text{curr} \gets \text{rep}(\text{ROOT})$
	\STATE $B(m , r_{\text{in}})$ is the inner ball associated with $\text{ROOT}$.
	\IF {$\sqrt{\sbreg}(q,\text{curr}) < D_{\text{near}}$}
		\STATE $\text{best}_q \gets \text{rep}(\text{ROOT})$
		\STATE $D_{\text{near}} \gets \sqrt{\sbreg}(q,\text{best}_q)$
	\ENDIF
	\IF {$\sqrt{\sbreg}(q, \text{curr}) < (1 + \frac{t}{2}) r_{\text{in}}$}
		\STATE $\text{ROOT} \gets \text{INNER CHILD}$
	\ELSE
		\STATE $\text{ROOT} \gets \text{OUTER CHILD}$
	\ENDIF
	\ENDWHILE
	\RETURN $\text{best}_q$
  \end{algorithmic}
  \label{roughNN}
\end{algorithm}

Let $\text{best}_q$ be the best candidate for nearest neighbor to $q$ 
found so far and $D_{\text{near}} = \sqrt{\sbreg}(\text{best}_q, q)$. Let $\text{nn}_q$ 
be the exact nearest neighbor to $q$ from point set $P$ and $D_{\text{exact}} = \sqrt{\sbreg}(\text{nn}_q,q)$ 
be the exact nearest neighbor distance. Finally, let $\textbf{curr}$ be the tree node currently
 being examined by our algorithm, and $\text{rep}_{\text{curr}}$ be a representative 
point $p \in P$ of $\textbf{curr}$. By convention $r_v$ represents the radius of the \emph{inner} $\sqrt{\sbreg}$-ball associated with a node $\textbf{v}$,
 and within each node $\textbf{v}$ we store $\text{rep}_v = m_v$, which 
is the center of $B_{\text{in}} (m_v , r_v)$. The node associated with the inner $\sqrt{\sbreg}$-ball $B_{\text{in}}$ 
is denoted by $\mathbf{v_{\text{in}}}$ and the node associated with $B_{\text{out}}$ is denoted by 
$\mathbf{v_{\text{out}}}$.

\begin{lemma}\label{ringsearch}
Given a $t$-ring tree $T$ for a point set with respect to $\sqrt{\sbreg}$ in a $\mu$-defective domain, where $t \leq \frac{1}{\log n}$ 
and query point $q$ we can find a $O(\mu + \frac{2 \mu^2 }{t})$ nearest neighbor to $q$ in $O( (\mu+1)^d d^{\frac{d}{2}} \log n)$ time.
\end{lemma}
\begin{proof}
Our search algorithm is a binary tree search. Whenever we reach node $\textbf{v}$, 
if $D(\text{rep}_v,q) < D_{\text{near}}$ set $\text{best}_q =\text{rep}_v$ and $D_{\text{near}} = \sqrt{\sbreg}(\text{rep}_v,q)$ 
as our current nearest neighbor and nearest neighbor distance respectively.
Our branching criterion is that if  $\sqrt{\sbreg}(\text{rep}_v, q) < (1 + \frac{t}{2}) r_v$,
 we continue search in $\mathbf{v_{\text{in}}}$, else we continue the  search in $\mathbf{v_{\text{out}}}$.
 Since the depth of the tree is $O(\log n)$ by Lemma \ref{ringsep}, this process will take $O(\log n)$ time.

Turning now to quality, let $\textbf{w}$ be the first node such that $\text{nn}_q \in \mathbf{w_{\text{in}}}$ but we 
searched in $\mathbf{w_{\text{out}}}$, or vice-versa.  After examining $\text{rep}_w$, $D_{\text{near}} \leq \sqrt{\sbreg}(\text{rep}_w, q)$ and $D_{\text{near}}$ can only decrease at each step.
 An upper bound on $\sqrt{\sbreg}(q, \text{rep}_w)/\sqrt{\sbreg}(q,\text{nn}_q)$ yields a bound on the quality of the approximate nearest neighbor produced.
In the first case, suppose $\text{nn}_q \in \mathbf{w_{\text{in}}}$, but we searched in $\mathbf{w_{\text{out}}}$.
Then $\sqrt{\sbreg}(\text{rep}_w, q) > \left( 1 + \frac{t}{2} \right) r_w$ and  $\sqrt{\sbreg}(\text{rep}_w, \text{nn}_q) < r_w$.
Now $\mu$-defectiveness implies that 
\begin{align*}
\mu \sqrt{\sbreg}(q,\text{nn}_q) &> \sqrt{\sbreg} (\text{rep}_w, q) - \sqrt{\sbreg}(\text{rep}_w, \text{nn}_q)  \\ 
\mu \sqrt{\sbreg}(q,\text{nn}_q) &> \left( 1+ \frac{t}{2} \right) r_w - r_w  \\
\sqrt{\sbreg}(q,\text{nn}_q) &>  \frac{t}{2 \mu} r_w.
\end{align*}
And for the upper bound on $\sqrt{\sbreg}(\text{rep}_w,q)/\sqrt{\sbreg}(q,\text{nn}_q)$, we again apply $\mu$-defectiveness to conclude that $\sqrt{\sbreg}(\text{rep}_w,q) - \sqrt{\sbreg}(q, \text{nn}_q) < \mu \sqrt{\sbreg}(\text{nn}_q, \text{rep}_w)$, which yields 

\begin{align*}
\frac{\sqrt{\sbreg}(\text{rep}_w,q)}{\sqrt{\sbreg}(q,\text{nn}_q)} &< 1 + \mu  \frac{r_w}{\sqrt{\sbreg}(q, \text{nn}_q)} \\ 
&< 1 + \mu \frac{r_w}{\frac{t}{2 \mu}r_w} \\
& = 1 + 2 \frac{\mu^2}{t}
\end{align*}.

We now consider the other case. Suppose $\text{nn}_q \in \mathbf{w_{\text{out}}}$ and we search 
in $\mathbf{w_{\text{in}}}$ instead.  By construction we must have
 $\sqrt{\sbreg}(\text{rep}_w, q) < \left( 1 + \frac{t}{2} \right) r_w$ and $\sqrt{\sbreg}(\text{rep}_w, \text{nn}_q) > ( 1 + t) r_w$.
Again, $\mu$-defectiveness yields $\sqrt{\sbreg}(q,\text{nn}_q) > \frac{t}{2 \mu} r_w$. Now we can simply take the ratios of the two:
$\frac{\sqrt{\sbreg}(\text{rep}_w,q)}{\sqrt{\sbreg}(q, \text{nn}_q)} < \frac{(1 + \frac{t}{2}) r_w}{\frac{t}{2 \mu} r_w } = \mu + \frac{2 \mu}{t}$. Taking an upper bound of the approximation provided by each case, the ring tree 
provides us a $\mu + 2 \frac{\mu^2}{t}$ approximation. The space/running time bound follows from Lemma \ref{ringstorage}, and noting that taking a thinner ring
($t \leq \frac{1}{\log n}$)  in the proof there only decreases the depth of the tree due to lesser duplication of points.
\end{proof}

\begin{corollary}
Setting $t = \frac{1}{\log n}$, given a point set with respect to $\sqrt{\sbreg}$ in a $\mu$-defective domain we can find a $O(\mu + 2 \mu^2 \log n)$ approximate nearest neighbor to a query point 
$q$ in $O(d^{\frac{d}{2}} (\mu+1)^d \log(n))$ time, using a $O(\frac{1}{\log n})$ ring separator tree constructed in
 $O(d^{\frac{d}{2}} (\mu+1)^d n \log(n))$ expected time.
\end{corollary}

\begin{proof}
 The query time is bounded by the depth of the tree, which is bounded in Lemma \ref{ringsep} . That we can construct
 a ring of our desired thickness at each step in reasonable expected time is guaranteed by \ref{improvedRing} . The space guarantee comes from Lemma \ref{ringstorage} and 
the quality of nearest neighbor obtained from our ring tree analyzed by Lemma \ref{ringsearch} .
Note that we are slightly abusing notation in Lemma \ref{improvedRing}, in that the separating ring
obtained there and which we use is not empty of points of $P$ as originally stipulated. However remember that if $\text{nn}_q$ is in the ring, then $\text{nn}_q$ repeats in \emph{both} children 
and cannot fall off the search path. Hence we can ``pretend'' the ring is empty  as in our analysis in 
Lemma \ref{ringsearch}.
\end{proof}

\section{Computing a $1+\eps$ approximation.}
\label{sec:finalized-algorithm}
We give now our overall algorithm for obtaining a $1 + \eps$ nearest neighbor in $O\left( \frac{1}{\eps^d}\log^{2d} n \right)$ query time under $\sqrt{\sbreg}$. We note that although our bounds are for $\sqrt{\sbreg}$, similar bounds follow in the same manner for any decomposable symmetric distance measure $D$, which satisfies an RTI and for which the ratio of diameter to side
length of a cube is bounded by $O(d^{O(1)})$. 

\subsection{Preprocessing}
We first construct an improved ring-tree $R$ on our point set $P$ in $O(n \log n)$ time as 
described in Lemma \ref{ringstorage}, with ring thickness $O(\frac{1}{\log n})$. 
We then compute an efficient orthogonal range reporting data structure on $P$ in $O(n \log ^{d-1} n)$ time,
 such as that described in~\cite{rangesearching} by Afshani \etal. We note the main result we need: 

\begin{lemma}\label{rangesearch}
We can compute a data structure from $P$ with $O(n \log ^{d-1} n)$ storage (and same construction time), such that given an arbitrary axis parallel box $H$ we can determine in $O(\log^{d} n)$ query time a point $p \in P \cap H$ if $|P \cap H| > 0$
\end{lemma} 

\subsection{Query handling}
Given a query point $q$, we use $R$ to obtain a point $q_{\text{rough}}$ in $O(\log n)$ time such that $D_{\text{rough}} =  \sqrt{\sbreg}(q, q_{\text{rough}}) \leq (1 + \mu^2 \log n) \sqrt{\sbreg}(q, \text{nn}_q)$. 
Given $q_{\text{rough}}$, we can use Lemma \ref{sballcover} to find a 
family $F$ of $2^d$ $\sqrt{\sbreg}$-cubes of side-length exactly $D_{\text{rough}}$ such that
 they cover the $\sqrt{\sbreg}$-ball $B(q, D_{\text{rough}})$. We use our range reporting structure to find a 
point $p \in P$ for all non-empty cubes in $F$ in a total of $2^d \log^d n$ time.
 These points act as representatives of the $\sqrt{\sbreg}$-cubes for what follows. Note that $\text{nn}_q$ must 
necessarily be in one of these $\sqrt{\sbreg}$-cubes, and hence there must be a ($1+\eps$)-nearest 
neighbor $q_{\text{approx}} \in P$ in some $G \in F$. To locate this $q_{\text{approx}}$, we 
construct a quadtree~\cite[Chapter 11]{snotes}~\cite{skipquadtrees} for repeated bisection and 
search on each $G \in F$.
 
Algorithm~\ref{algo} describes the overall procedure. 
We call the collection of all cells produced during the procedure a \emph{quadtree}.
We borrow the presentation in Har-Peled's book~\cite{snotes} with the important qualifier 
that we construct our quadtree at runtime. The terminology here is as introduced earlier in 
Section \ref{sec:ringsec}.

\begin{algorithm}
  \caption{QueryApproxNN$(P,\mathbf{root},q)$}
  \begin{algorithmic}
    \STATE Instantiate a queue $Q$ containing all cells from $F$ along with their representatives and 
           enqueue \textbf{root}.
    \STATE Let $D_{\text{near}} =\sqrt{\sbreg}(\text{rep}_{\text{root}}, q)$,
    $\text{best}_q = \text{rep}_{\mathbf{root}}$
    \REPEAT 
    \STATE Pull off the head of the queue and place it in $\mathbf{curr}$.
    \IF{$\sqrt{\sbreg}(\text{rep}_{\text{curr}}, q) < \sqrt{\sbreg}(\text{best}_q, q)$}
      \STATE Let $\text{best}_q = \text{rep}_{\text{curr}}$, 
             $D_{\text{near}} = \sqrt{\sbreg}(\text{best}_q, q)$
      \STATE Bisect $\textbf{curr}$ according to procedure of Lemma   
             \ref{bisectionProcedure}; denote the result as $\{G_i\}$.
      \FORALL{$G_i$}
        \STATE As described in \ref{bisectionProcedure}, check if $G_i$ is non-empty by passing it to our range reporting structure, which will also return us some $p \in P$  if $G_i$ is not empty. 
        \STATE Also check if $G_i$  may contain a point closer than $(1- \frac{\eps}{2}) D_{\text{near}}$ to $q$. (This may be done in $O(d)$ time for each cell, given the coordinates of the corners.) 
        \IF{$G_i$ is non-empty AND has a close enough point to $q$}
          \STATE Let $\text{rep}_{G_i} = p$
          \STATE Enqueue $G_i$
        \ENDIF
      \ENDFOR
    \ENDIF
    \UNTIL $Q$ is empty
    \STATE Return $\text{best}_q$
  \end{algorithmic}
  \label{algo}
\end{algorithm}
%First we demonstrate that our algorithm always returns a $(1+ \eps)$-nearest neighbor to $q$ correctly.
\begin{lemma}\label{correctness}
Algorithm~\ref{algo} will always return a $(1+\eps)$-approximate nearest neighbor. 
\end{lemma}
\begin{proof}
Let $\text{best}_q$ be the point returned by the algorithm at the end of execution. By the method of the algorithm, for all points $p$ for which the distance is directly evaluated, we have that
$
\sqrt{\sbreg}(\text{best}_q, q) < \sqrt{\sbreg}(p,q) 
$. 
The terminology here is as in Section \ref{sec:ringsec}.
We look at points $p$ which are \emph{not} evaluated during the running of the algorithm, 
i.e. we did not expand their containing cells. But by the criterion of the algorithm for not expanding a cell, it must be that $\sqrt{\sbreg}(\text{best}_q, q) (1 - \frac{\eps}{2}) < \sqrt{\sbreg}(p,q)$. For $\eps <1$, this means that $(1+ \eps)\sqrt{\sbreg}(p,q)  > \sqrt{\sbreg}(\text{best}_q,q)$ for any $p \in P$, including $\text{nn}_q$. So $\text{best}_q$ is indeed a $1 + \eps$ approximate nearest neighbor.
\end{proof}

We must analyze the time complexity of a single iteration of our algorithm, namely the complexity of a subdivision of a $\sqrt{\sbreg}$-box $G$ and determining which of the $2^d$ $\sqrt{\sbreg}$-subcells of $G$ are non-empty.
\begin{lemma}\label{bisectionProcedure}
Let $G$ be a $\sqrt{\sbreg}$-box with maximum $\sqrt{\sbreg}$-side-length $s$ and $G_i$ its subcells produced by bisecting along each side of $G$ under $\sqrt{\sbreg}$. For all non-empty $\sqrt{\sbreg}$-subcubes $G_i$ of $G$, we can find $p_i \in P \cap G_i$ in $O(2^d \log^{d} n)$ total time complexity, and the maximum $\sqrt{\sbreg}$-side-length of any $G_i$ is at most $\frac{s}{2}$.
\end{lemma}

\begin{proof}
 Note that $G$ is defined as a product of $d$ intervals. For each interval, we can find an approximate bisecting point under
$\sqrt{\sbreg}$ in $O(1)$ time and by the RTI each subinterval is of length at most $\frac{s}{2}$ under $\sqrt{\sbreg}$. This leads to an $O(d)$ cost to find a bisection point for all intervals, which define $O(2^d)$ $\sqrt{\sbreg}$-subboxes or children of $G$.

We pass each $\sqrt{\sbreg}$-subbox of $G$ to our range reporting structure. By lemma \ref{rangesearch}, this takes $O(\log^{d} n)$ time to check emptiness or return a point $p_i \in P$ contained in the child, if non-empty. Since there are $O(2^d)$ non-empty children of $G$, this implies a cost of $2^d(\log^{d} n)$ time incurred.

 Checking each of the non-empty subboxes $G_i$ to see if it may contain a point closer than $(1 - \frac{\eps}{2}) D_{\text{near}}$ to $q$ takes a further $O(d)$ time per cell or $O(d 2^d)$ time. 
\end{proof}

We now bound the number of cells that will be added to our search queue. We do so indirectly, by placing a lower bound on the maximum $\sqrt{\sbreg}$-side-length of all such cells.

\begin{lemma}\label{depthCube}
Algorithm \ref{algo} will not add the children of node  $\textbf{C}$ to our search queue 
if the maximum side-length of $\textbf{C}$ is less than $\frac{\eps \sqrt{\sbreg}(q, \text{nn}_q)}{2 \mu \sqrt{d}}$.
\end{lemma}

\begin{proof}

Let $\Delta(\textbf{C})$ represent the $\sqrt{\sbreg}$-diameter of cell $\textbf{C}$. By construction, 
we can expand $\textbf{C}$ only if some subcell of $\textbf{C}$ contains a point $p$ 
such that $\sqrt{\sbreg}(p , q) \leq (1 - \frac{\eps}{2}) D_{\text{near}}$. Note that since $\textbf{C}$ is 
examined, we have $D_{\text{near}} \leq \sqrt{\sbreg}(\text{rep}_C , q)$. 
Now assuming we expand $\textbf{C}$, then we must have:
\begin{equation}
\mu \Delta(\textbf{C}) > \sqrt{\sbreg}(\text{rep}_C,q) - \sqrt{\sbreg}(p, q) 
\geq D_{\text{near}} - (1 - \frac{\eps}{2})D_{\text{near}} 
= \frac{\eps}{2} D_{\text{near}} \end{equation}

So $\eps/(2 \mu) D_{\text{near}} < \Delta(\textbf{C})$.
First note $\sqrt{\sbreg}(\text{rep}_C, q) < D_{\text{near}} $. Also, by definition,  $\sqrt{\sbreg}(q, \text{nn}_q) < D_{\text{near}}$. And $\Delta(\textbf{C}) < \sqrt{d} s$ by lemma \ref{diamtolen} where $s$ is the maximum side-length of $\textbf{C}$. 
Making the appropriate substitutions yields us our required bound.
\end{proof}

Given the bound on quadtree depth (Lemma \ref{depthCube}), and using the fact that at most $2^{xd}$ nodes are expanded at level $x$, we have:
\begin{lemma}\label{timeFinal}
Given a $\sqrt{\sbreg}$-cube $G$ of $\sqrt{\sbreg}$-side-length $D_{\text{rough}}$, we can compute a $(1 + \eps)$-nearest neighbor to $q$ in $O\left(  \frac{1}{\eps^d} 2^d \mu^d d^{\frac{d}{2}} \left( \frac{D_{\text{rough}}}{\sqrt{\sbreg}(q, \text{nn}_q)} \right)^d  \log^d n \right )$ time.
\end{lemma}

\begin{proof}
Consider a quadtree search from $q$ on a $\sqrt{\sbreg}$-cube $G$ of $\sqrt{\sbreg}$-side-length $D_{\text{rough}}$. By lemma \ref{depthCube}, our algorithm will not 
expand cells with all $\sqrt{\sbreg}$-side-lengths smaller than  $\frac{  \eps \sqrt{\sbreg}(q, \text{nn}_q)}{2 \mu \sqrt{d}}$. But since the $\sqrt{\sbreg}$-side-length reduces by at least half in each dimension upon each split, all $\sqrt{\sbreg}$-side-lengths are less than this value after 
$x =\log \left( D_{\text{rough}}/\frac{  \eps  \sqrt{\sbreg}(q, \text{nn}_q)}{2 \mu \sqrt{d}} \right)$ repeated bisections of our root cube.

Noting that $O(\log^d n)$ time is spent at each node by lemma \ref{bisectionProcedure}, and that at the $x$th level the 
number of nodes expanded is $2^{xd}$, we get a final time complexity bound of 
$O\left( \frac{1}{\eps^d} 2^d \mu^d d^{\frac{d}{2}} \left( \frac{D_{\text{rough}}}{\sqrt{\sbreg}(q, \text{nn}_q)} \right)^d \log^d n \right)$. 
\qedhere
\end{proof}

Substituting $D_{\text{rough}} = \mu^2 \log n \sqrt{\sbreg}(q, \text{nn}_q) $ in Lemma \ref{timeFinal} gives us a bound 
of $O\left(2^d \frac{1}{\eps^d}  \mu^{3d} d^{\frac{d}{2}} \log^{2d} n\right)$. 
 This time is per $\sqrt{\sbreg}$-cube of $F$ that covers $B(q,D_{\text{rough}})$. Noting 
that there are $2^d$ such $\sqrt{\sbreg}$-cubes gives us a final time complexity of $O\left(2^{2d}   \frac{1}{\eps^d} \mu^{3d} d^{\frac{d}{2}} \log^{2d} n \right)$.
For the space complexity of our run-time queue, observe that the number of nodes in our queue 
increases only if a node has more than one non-empty child, i.e, there is a split of our $n$ points. 
Since our point set may only split $n$ times, this gives us a bound of $O(n)$ on the space complexity 
of our queue.

\section{Logarithmic bounds, with further assumptions.}
\label{sec:condition}
For a given $\sbreg$, let $c_0 = \max_{ i\in [1..d]} \sqrt{\frac{ \max_x \phi_i'' (x)} {  \min_y \phi_i''(y)}}$ 
over our bounded subset of the domain ($c_0$ may be infinity over the unrestricted domain, or on a subset over whose closure $\phi''$ tends to infinity or zero). $c_0$ is susceptible to the choice of bounded subset of the domain and in general grows as we expand our allowed subset.
We conjecture that $c_0 = \Theta(\mu)$ although we cannot prove it. In particular, we show that if we assume a bounded $c_0$ (in addition to $\mu$), 
we can obtain a $1 + \eps$ nearest neighbor in time $O(\log n + (\frac{1}{\eps})^d)$ time for $\sqrt{\sbreg}$. We do so by constructing a \emph{Euclidean} quadtree $T$ on our set
in preproccessing and using $c_0$ and $\mu$ to express the bounds obtained in terms of $\sqrt{\sbreg}$.

We will refer to the Euclidean distance $l_2$ as $D_e$ and note first the following key relation between $\sqrt{\sbreg}$ and $D_e$, where $c_0$ serves to relate the two measures by some constant factor. Nock \etal~\cite{mixed} use a comparable measure to $c_0$ as do Sra \etal~\cite{tensorclust}, for similar purposes of establishing a constant factor approximation 
to the Euclidean distance.

\begin{lemma}\label{EucBregBisect}
Suppose we are given a interval $I = [x_1 x_2] \subset \reals$ s.t. $x_1 < x_2$, $D_e(x_1, x_2) = r_e$, and $\sqrt{\sbreg(x_1,x_2)} = r_{\phi}$. 
Suppose we divide $I$ into $m$ subintervals of equal length with endpoints $x_1 = a_0, a_1,\ldots a_{m-1}, a_m = x_2$, where $a_i < a_{i+1}$ and $D_e (a_i, a_{i+1}) = r_e /m$, $\forall i\in[0..m-1]$.
Then $\frac{r_{\phi} }{c_0 m } \leq \sqrt{\sbreg (a_i, a_{i+1})} \leq  \frac{c_0 r_{\phi}}{m}$.
\end{lemma}

\begin{proof}
We can relate $\sqrt{\sbreg}$ to $D_e$ via the Taylor expansion of $\sqrt{\sbreg}$:
 $\sqrt{\sbreg(a,b)} = \sqrt{ \phi''(\bar{x}) } D_e(a,b)$ for some $\bar{x} \in [a b]$. Combining this with $c_0$ yields 

\begin{equation}
\frac{ \min_{i} \sqrt{\sbreg(a_i , a_{i+1})}}{\sqrt{\sbreg(x_1,x_2) } } 
\geq \frac{D_e(a_i, a_{i+1})}{c_0 D_e(x_1, x_2)} = \frac{1}{c_0 m}
\end{equation}

and  
\begin{equation}
\frac{ \max_{i} \sqrt{\sbreg(a_i , a_{i+1})}}{\sqrt{\sbreg(x_1,x_2) } } 
\leq c_0 \frac{D_e(a_i, a_{i+1})}{D_e(x_1, x_2)} = \frac{c_0}{m}.
\end{equation}

\end{proof}

\begin{corollary}\label{repbis}
If we recursively bisect an interval $I = [x_1 x_2] \subset \reals$ s.t. $D_e(x_1, x_2) = r_e$ and $\sqrt{\sbreg(x_1,x_2)} = r_{\phi}$ 
into $2^i$ equal subintervals (under $D_e$), then 
$\frac{r_{\phi} }{c_0 2^i } \leq \sqrt{\sbreg (a_k, a_{k+1}) }\leq  \frac{c_0 r_{\phi}}{2^i}$ for any of the
subintervals $[a_k a_{k+1}]$ so obtained. Hence after $\log \frac{c_0 r_{\phi}}{ x}$ subdivisions, all intervals will be of length 
at most $x$ under $\sqrt{\sbreg}$.
Also, given a cube of initial side-length $r_{\phi}$, after $\log \frac{c_0 r_{\phi} }{ x}$ repeated bisections (under $D_e$) 
the diameter will be at most  $\sqrt{d} x$ under $\sqrt{\sbreg}$.
\end{corollary}

We find the smallest enclosing $\sqrt{\sbreg}$-cube $C$ that bounds our point set, 
and then construct our compressed Euclidean quadtree in preprocessing on this cube.  Say $C$
is of side-length $s$.
Corollary \ref{repbis} gives us that for cells formed at the $i$-th level of decomposition, 
the side-length under $\sqrt{\sbreg}$ is between $\frac{s}{c_0 2^i}$ and $\frac{c_0 s}{ 2^i}$. Refer to these cells formed at the $i$-th level as $L_i$.

%We present the key lemmas for our analysis.
\begin{lemma}\label{packingTheTree}
Given a $\sqrt{\sbreg}$-ball $B$ of radius $r$, let $i = \log \frac{ s}{c_0 r}$. 
Then $|L_i \cap B| \leq O(2^d)$ and the side-length of each cell in $L_i$ is between $r$ and ${c_0}^2 r$ under $\sqrt{\sbreg}$. 
We can also explicitly retrieve the quadtree cells corresponding to $|L_i \cap B|$ in $O(2^d \log n)$ time.
\end{lemma}

\begin{proof}
Note that for cells in $L_i$, we have side-lengths under $\sqrt{\sbreg}$ between $\frac{s}{c_0 2^i}$ and $\frac{c_0 s}{2^i}$ by Corollary \ref{repbis}.
 Substituting $i = \log \frac{s}{c_0 r}$,  these cells have side-length between $r$ and ${c_0}^2 r$ under $\sqrt{\sbreg}$. 
By the reverse triangle inequality and Lemma \ref{1dintersect}, we get our required bound for $|L_i \cap B|$.
In preconstruction of our quadtree $T$ we maintain for each dimension the corresponding interval quadtree $T_k$, $\forall k\in [1..d]$. 
Observe this incurs at most $O(n)$ storage, with $d$ in the big-Oh.  For retrieving the actual cells $|L_i \cap B|$, we first
find the $O(1)$ intervals from level $i$ in each $T_k$ that may intersect $B$. Taking a product of these, we get $O(2^d)$ cells 
which are a superset of the canonical cells $L_i \subset T$. Each cell may be looked up in $O(\log n)$ time from the compressed
 quadtree \cite{snotes} so our overall retrieval time is $O (2^d \log n)$.
\end{proof}

 Given query point $q$, we first obtain in $O(\log n)$ time with our ring-tree a rough $O(n)$ ANN $q_{\text{rough}}$ 
under $\sqrt{\sbreg}$ s.t. \\ $D_{\text{rough}} = \sqrt{\sbreg(q, q_{\text{rough}})} = \mu^2 n \sqrt{\sbreg (q, \text{nn}_q)}$.
Note that we can actually obtain a $O(\log n)$-ANN instead, using the results of Section \ref{ringsep}. But a coarser approximation
of $O(n)$-ANN suffices here for our bound. The tree depth (and implicitly the storage and running time) is still bounded
by the $O(d^{\frac{d}{2}} (\mu+1)^d \log n)$ of Lemma \ref{ringsep}, since in using thinner rings we have less point duplication
and the same proportional reduction in number of points in each node at each level.

Now Lemma \ref{packingTheTree}, we have $O(2^d)$ quadtree cells intersecting  $B (q, \sqrt{\sbreg (q, q_{\text{rough}})} )$.

 Let us call this collection of cells $Q$.  We then carry out a quadtree search on each element of $Q$. Note that
 we expand only cells which may contain a point nearer to query point $q$ than the current best candidate.  
We bound the depth of our search using $\mu$-defectiveness similar to Lemma~\ref{depthCube}:

\begin{lemma}\label{modDepth}
We will not expand cells of $\sqrt{\sbreg}$-diameter less than \\$\frac{\eps  \sqrt{\sbreg(q, \text{nn}_q)}}{2 \mu}$ or cells whose all side-lengths w.r.t. $\sqrt{\sbreg}$ are less than $\frac{\eps  \sqrt{\sbreg(q, \text{nn}_q)}}{2 \mu \sqrt{d}}$.
\end{lemma}

For what follows, refer to our \emph{spread} as $\beta =  \frac{D_{\text{rough}}}{\sqrt{\sbreg(q,\text{nn}_q)}}$.

\begin{lemma}\label{treeDepth}
We will only expand our tree to a depth of \\$k= \log (2 {c_0}^3 \mu \beta \sqrt{d}/\eps)$.
\end{lemma}
\begin{proof}
 Using Lemma \ref{modDepth} and Corollary \ref{repbis}, each cell of $Q$ will be expanded only to a depth of 
$k= \log \left( c_0 {c_0}^2  D_{\text{rough}} 
/   \frac{\eps \sqrt{\sbreg(q, \text{nn}_q)}}{2 \mu \sqrt{d}} \right)$. 
This gives us a depth of $\log (2 {c_0}^3 \mu \beta \sqrt{d}/\eps)$.
\end{proof}

\begin{lemma}\label{breathnum}
The number of cells examined at the $i$-th level is  $n_i  < 2^d \left( \mu^d d^{\frac{d}{2}}  c_0^{4d}+ (\frac{2^i c_0}{\beta})^d \right)$.
\end{lemma}

\begin{proof}
Recalling that the cells of $Q$ start with side-length at most $c_0^2 D_{\text{rough}}$ under $\sqrt{\sbreg}$, at the $i$-th level the $\sqrt{\sbreg}$-diameter of cells is at most $\frac{c_0^3 \sqrt{d} D_{\text{rough}}}
{2^i}$ , by Corollary~\ref{repbis}. 
Hence by $\mu$-defectiveness, there must be some point examined by our algorithm at $\sqrt{\sbreg}$-distance at most
$D_{\text{best}} =\sqrt{\sbreg}(q,\text{nn}_q) + \frac{\mu c_0^3 \sqrt{d} D_{\text{rough}}}{2^i}$. Note that our algorithm will
only expand cells within this distance of $q$.

The $\sqrt{\sbreg}$ side-length of a cell $\textbf{C}$ at this level is at least $\Delta(\textbf{C}) = \frac{ D_{\text{rough}}}{c_0 2^i}$. Applying the
packing bounds from Lemma \ref{scubeCover}, and the fact that $(a+b)^d < 2^d(a^d + b^d)$, the number of cells expanded is at most
\[
n_i 
= \left( \frac{D_{\text{best}}}{\Delta(\textbf{C})} \right)^d 
%= \left( \mu \sqrt{d} c_0^4 +  \frac{c_0 2^i}{\beta} \right)^d 
< 2^d \left( \mu^d d^{\frac{d}{2}} c_0^{4d} + \left( \frac{c_0 2^i}{\beta} \right)^d \right).
\qedhere
\]
\end{proof}
Finally we add the 
%$\sum_{i = 1}^{i= \log (2 c_0^3 \mu \beta \sqrt{d}/\eps)} n_i$ 
$n_i$ to get the total number of nodes explored: 
\[ 
\sum_i n_i = O \left( 2^d \mu^d d^{\frac{d}{2}} c_0^{4d} \log (2 {c_0}^3 \mu \beta \sqrt{d}/\eps) + 2^{2d} c_0^{4d} \mu^d d^{\frac{d}{2}}/\eps^d \right).
\]
Recalling that $\beta =\frac{D_{\text{rough}}}{\sqrt{\sbreg}(q,\text{nn}_q)} = \mu^2 n$, substituting back and ignoring lower order terms, the time complexity is 
\[ 
O \left(2^d  \mu^d d^{\frac{d}{2}} c_0^{4d} \log n + 2^{2d} c_0^{4d} \mu^d d^{\frac{d}{2}}/\eps^d \right).
\]

Accounting for the $2^d$ cells in $Q$ that we need to search, this adds a further $2^d$ multiplicative factor.
This time complexity of this quadtree phase(number of cells explored) of our algorithm dominates the time complexity of the ring-tree 
search phase of our algorithm, and hence is our overall time
complexity for finding a $(1+ \eps)$ ANN to $q$.
For space and pre-construction time, we note that compressed Euclidean quadtrees can be built in $O(n \log n)$ time and require $O(n)$ space~\cite{snotes},
which matches our bound for the ring-tree construction phase of our algorithm requiring $O(n \log n)$ time and $O(n)$ space. 

\section{The General Case: Asymmetric Divergences}\label{sec:generalizations}

Without loss of generality we will focus on the \emph{right-sided} nearest neighbor: given a point set $P$, query point $q$ and $\epsilon \ge 0$, find $x \in P$ that approximates $\min_{p \in P} D(p, q)$ to within a factor of $(1+\epsilon)$.  Since a Bregman divergence is not in general $\mu$-defective, we will consider instead $\sqrt{\breg}$: by monotonicity and with an appropriate choice of $\epsilon$, the result will carry over to $\breg$. 

We list  three issues that have to be resolved to complete the algorithm. Firstly, because of asymmetry, we cannot bound the diameter of a 
quadtree cell $\textbf{C}$ of side-length $s$ by $s\sqrt{d}$. However, as the proof of Lemma \ref{acube} shows, we can 
choose a \emph{canonical corner} of a cell such that a (directed) ball of radius $s\sqrt{d}$ centered at that corner covers the cell. 
By $\mu$-defectiveness, we can now conclude (see lemma \ref{adiamtolength}) that the diameter of $\mathbf{C}$ is at most $(\mu+1)s\sqrt{d}$ (note that this incurs an extra 
factor of $\mu+1$ in all expressions). Secondly, since while $\sqrt{\breg}$ satisfies $\mu$-defectiveness (unlike $\breg$) the opposite is 
true for the reverse triangle inequality, which is satisfied by $\breg$ but not $\sqrt{\breg}$. This requires the use of a weaker packing 
bound based on Lemma \ref{1dsqrtbregint}, introducing dependence in $1/\epsilon^2$ instead of $1/\epsilon$. And thirdly, the lack of symmetry means
we have to be careful of the use of directionality when proving our bounds. Perhaps surprisingly, the major part of the arguments
carry through simply by being consistent in the choice of directionality.

 Note that for this section  we are referring to $\sqrt{\breg}$.
 With some small adjustments, similar bounds can be obtained for
more generic asymmetric, monotone, decomposable and $\mu$-defective distance $D$ measures satisfying packing bounds.
The left-sided asymmetric nearest neighbor can be determined analogously.

Finally, given a bounded domain $M$, we have that $\sqrt{\breg}$  is left-sided $\mu$-defective 
for some $\mu_L$ and right sided $\mu$-defective for some $\mu_R$ (see Lemma \ref{Arootmubreg} for detailed proof). 
For what follows, let $\mu = \max(\mu_L, \mu_R)$  and describe $M$ as simply $\mu$-defective. 
%Note that for symmetric distance measures, $\mu_L = \mu_R$.

Most of the proofs here mirror their counterparts in Sections \ref{sec:ringsec} and \ref{sec:finalized-algorithm}. 

\subsection{Asymmetric ring-trees}
\label{subsec:ringextension} 
Since we focus on \emph{right}-near-neighbors, all balls and ring separators referred to will use \emph{left-balls}
 i.e balls $B(m,r) = \{x \mid  D(m,x) < r\}$. As in Section \ref{sec:ringsec}, we will design a ring-separator algorithm and 
use that to build a ring-separator tree.

\begin{lemma} \label{leftcircle}
Let $D$ be a $\mu$-defective distance, and let $B(m,r)$ be a left-ball with respect to $D$. Then for any two points $x,y \in B(m,r)$, $D(x,y) < (\mu+1) r$.
\end{lemma}

 \begin{proof}
 Follows from the definition of right sided $\mu$-defectiveness.
  \begin{align*}
  D(x,y) - D(m,y) &< \mu  D(m,x) 
  \\ D(x,y) <  \mu r +  D(m,y)  &= (\mu + 1) r
  \end{align*}
 \end{proof}
 
 \begin{corollary}
 For any $\sqrt{\breg}$-left-ball $B(m,r)$ and two points $x,y \in B(m,r)$,
 $\sqrt{\breg}(x,y) < (\mu+1) r$. 
 \end{corollary}
\begin{proof}
Since $\sqrt{\breg}$ is $\mu$-defective over a prespecified restricted domain.
\end{proof}

As in Lemma \ref{randommufraction} we can construct  (in $O(nc)$ expected time) a $(\mu+1)$-approximate $\sqrt{\breg}$-left-ball 
enclosing $\frac{n}{c}$ points. This in turn yields a ring-separator construction, 
and from it a ring tree with an extra $(\mu+1)^d d^{\frac{d}{2}}$ factor in depth as compared to
symmetric ring-trees, due to the weaker packing bounds for $\sqrt{\breg}$.

 We note that the asymptotic bounds for ring-tree storage and construction time follow from purely combinatorial arguments
 and hence are unchanged for $\sqrt{\breg}$. 
Once we have the ring- tree, we can use it as before to identify a rough near-neighbor for a query $q$;
 once again, exploiting $\mu$-defectiveness gives us the desired approximation guarantee for the result.

 \begin{lemma} \label{randommufractionassym}
 Given any parameter $1 \leq c \leq n$, we can compute in $O(nc)$  randomized time a $\sqrt{\breg}$-left-ball $B(m,r')$ such that $r' \leq (\mu + 1) r_{\text{opt},c}$ and $B(m,r') \cap P \geq \frac{n}{c}$.
\end{lemma}  

\begin{proof}
 Follows identically to the proof of Lemma \ref{randommufraction}.
\end{proof}

\begin{lemma}\label{ringassym}
There exists a parameter $c$ (which depends only on $d$ and $\mu$), such that for any $d$-dimensional point set $P$ and any $\mu$-defective $\sqrt{\breg}$, we can find a $O(\frac{1}{\log n})$ left-ring separator $R_{P,c}$
in $O(n)$ expected time.
\end{lemma}
 
\begin{proof}
First, using our randomized construction, we compute a $\sqrt{\breg}$-left-ball $S=B(m,r_1)$ (where $m \in P$) containing $\frac{n}{c}$ 
points such that $r_1 \leq (\mu+1)r_{\text{opt},c}$, where $c$ is a parameter to be set. 
Consider the $\sqrt{\breg}$-left-ball $\bar{S} = B(m,2r_1)$.  As described in Lemma \ref{aballcover}, $\bar{S}$ can be covered by $2^d$ 
$\sqrt{\breg}$-hypercubes of side-length $2 r_1$,
 the union of which we shall refer to as $H$. Set $L = (\mu+1)\sqrt{d}$. 
 Imagine a partition of $H$ into a grid, where each cell is of side-length $\frac{r_1}{L}$. 
Each cell in this grid can be covered by a $\sqrt{\breg}$-ball of radius $\Delta(\frac{r_1}{L}, d) = \frac{r_1}{\mu+1} \leq r_{\text{opt},c}$
centered on it's lowest corner.  This implies any cell will contain at most $\frac{n}{c}$ points, 
by the definition of $r_{\text{opt},c}$. 

 By Lemma \ref{acubeCover} the grid on $H$ has at most $2^d(2r_1/\frac{r_1}{L})^{2d} = (4 (\mu+1) \sqrt{d})^{2d}$ cells.
 Each cell may contain at most $\frac{n}{c}$ points. In particular, set $c = 2(4 (\mu+1) \sqrt{d})^{2d}$. Then we have that $H$ may contain 
at most $\frac{n}{c} (4 (\mu+1) \sqrt{d})^{2d}  = \frac{n}{2}$ points,  or since $\bar{S} \subset H$, $\bar{S}$ contains at most $\frac{n}{2}$ points 
and $\bar{S}'$ contains at least $\frac{n}{2}$ points.  The rest of the proof goes through as in Lemma \ref{improvedRing}
\end{proof}

We proceed now to the construction of our ring-tree using the basic ring-separator structure of Lemma \ref{ringassym}.

\begin{lemma}\label{assymringsep}
Given any point set $P$, we can construct a $O(\frac{1}{\log n})$ 
left ring-separator tree $T$ under $\sqrt{\breg}$ of depth $O(d^d (\mu+1)^{2d} \log n)$ .
\end{lemma}

\begin{proof}
Repeatedly partition $P$ by Lemma \ref{ringassym} into $P^{v}_{\text{in}}$ and $P^{v}_{\text{out}}$ where $\textbf{v}$ is the parent node. Store only the single point $\text{rep}_v = m \in P$ in node $\textbf{v}$, the center of the ball $B(m,r_1)$. We continue this partitioning until we have nodes with only a single point contained in them.

Since each child contains at least $\frac{n}{c}$ points, each subset reduces by a factor of at least $1 - \frac{1}{c}$ at each step, 
and hence the depth of the tree is logarithmic. 
We calculate the depth more exactly, noting that in Lemma \ref{ringassym}, $c = O(d^d (\mu+1)^{2d})$.
Hence the depth $x$ can be bounded as:
\begin{align*}
& n (1 - \frac{1}{c})^x = 1 
\\ &(1- \frac{1}{c})^x = \frac{1}{n} 
\\ &x 
= \frac{ \ln \frac{1}{n}}{\ln (1 - \frac{1}{c})} 
= \frac{-1}{\ln (1 - \frac{1}{c})} \ln n
\\ & x \leq c \ln n 
 = O \left( d^d (\mu+1)^{2d} \log n \right)
\end{align*}
\end{proof}

Note that Lemma \ref{assymringsep} also serves to bound the query time of our data structure. We need only now bound the approximation quality.
The derivation is similar to Lemma \ref{ringsearch}, but with some care about directionality.

\begin{lemma}\label{ringsearchassym}
Given a $t$-ring tree $T$ for a point set with respect to a $\mu$-defective  $\sqrt{\breg}$, where $t \leq \frac{1}{\log n}$, and query point $q$ we can find a $O(\mu + \frac{2 \mu^2 }{t})$ nearest neighbor  to query point $q$ in $O( (\mu+1)^{2d} d^d \log n)$ time.
\end{lemma}

\begin{proof}
Our search algorithm is a binary tree search. Whenever we reach node $\textbf{v}$, if $\sqrt{\breg}(\text{rep}_v,q) < D_{\text{near}}$ 
set $\text{best}_q =\text{rep}_v$ and $D_{\text{near}} = \sqrt{\breg}(\text{rep}_v,q)$ as our current 
nearest neighbor and nearest neighbor distance respectively.  
Our branching criterion is that if  $\sqrt{\breg} (\text{rep}_v, q) < (1 + \frac{t}{2}) r_v$, we
 continue search in $\mathbf{v_{\text{in}}}$, else we continue the  
search in $\mathbf{v_{\text{out}}}$. Since the depth of the tree is $O(\log n)$ by Lemma \ref{assymringsep}, 
this process will take $O(\log n)$ time.
\begin{figure}[H]
  \begin{center}
    \includegraphics[scale = 0.4]{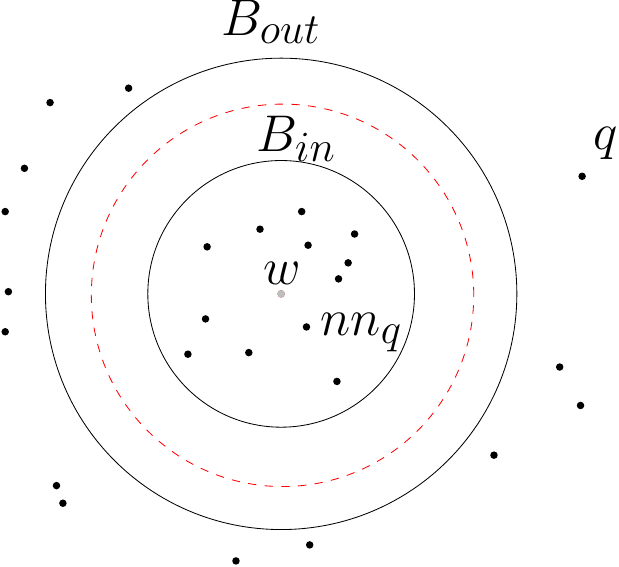}
  \end{center}
  \caption{$q$ is outside $(1+\frac{t}{2})r_{\text{in}}$ so we search $\mathbf{w_{\text{out}}}$, but $\text{nn}_q \in \mathbf{w_{\text{in} }}$}
  \label{case1}
\end{figure}
 Let $\textbf{w}$ be the first node such that $\text{nn}_q \in \mathbf{w_{\text{in}}}$ but we searched in $\mathbf{w_{\text{out}}}$, or vice-versa. 
The analysis goes by cases. In the first case as seen in figure \ref{case1}, suppose $\text{nn}_q \in \mathbf{w_{\text{in}}}$, but we searched in $\mathbf{w_{\text{out}}}$.
Then
\begin{align*}
\sqrt{\breg}(\text{rep}_w, q) &> \left( 1 + \frac{t}{2} \right) r_w 
\\ \sqrt{\breg}(\text{rep}_w, \text{nn}_q) &< r_w.
\end{align*}

Now left-sided $\mu$-defectiveness implies a lower bound on the value of $\sqrt{\breg}(\text{nn}_q,q)$:
\begin{align*}
\mu \sqrt{\breg}(\text{nn}_q, q) &> \sqrt{\breg} (\text{rep}_w, q) - \sqrt{\breg}(\text{rep}_w, \text{nn}_q) 
\\ \mu \sqrt{\breg}(\text{nn}_q,q) &> \left( 1+ \frac{t}{2} \right) r_w - r_w 
\\ \sqrt{\breg}(\text{nn}_q,q) &>  \frac{t}{2 \mu} r_w,
\end{align*}

And for the upper bound on $\sqrt{\breg}(\text{rep}_w,q)/\sqrt{\breg}(\text{nn}_q,q)$. First by right-sided $\mu$-defectiveness:
\begin{align*}
\sqrt{\breg}(\text{rep}_w,q) - \sqrt{\breg}(\text{nn}_q,q) &< \mu \sqrt{\breg}( \text{rep}_w, \text{nn}_q) 
\\ \sqrt{\breg}(\text{rep}_w,q)  &< \sqrt{\breg}(\text{nn}_q,q) + \mu  r_w
\\ \frac {\sqrt{\breg}(\text{rep}_w,q)}{\sqrt{\breg}(\text{nn}_q,q)} &< 1 + \mu  \frac{r_w}{\sqrt{\breg}( \text{nn}_q,q)} 
\\ \frac{\sqrt{\breg}(\text{rep}_w,q)}{\sqrt{\breg}(\text{nn}_q,q)} &< 1 + \mu \frac{r_w}{\frac{t}{2 \mu}r_w} 
\\ \frac{\sqrt{\breg}(\text{rep}_w,q)}{\sqrt{\breg}(\text{nn}_q,q)} &< 1 + \mu \frac{2 \mu}{t} 
\\ \frac{\sqrt{\breg}(\text{rep}_w,q)}{\sqrt{\breg}(\text{nn}_q,q)} &< 1 + 2 \frac{\mu^2}{t}
\end{align*}

We now consider the other case. Suppose $\text{nn}_q \in \mathbf{w_{\text{out}}}$ and we search in $\mathbf{w_{\text{in}}}$ instead. The analysis is almost identical. By construction we must have:
\begin{align*}
\sqrt{\breg}(\text{rep}_w, q) &< \left( 1 + \frac{t}{2} \right) r_w 
 \\\sqrt{\breg}(\text{rep}_w, \text{nn}_q) &> ( 1 + t) r_w 
\end{align*}

Again, left-sided $\mu$-defectiveness yields:
\begin{align*}
\sqrt{\breg}(\text{nn}_q,q) &> \frac{t}{2 \mu} r_w
\end{align*}

We can simply take the ratios of the two:
\begin{align*}
\frac{\sqrt{\breg}(\text{rep}_w,q)}{\sqrt{\breg}( \text{nn}_q,q)} &< \frac{(1 + \frac{t}{2}) r_w}{\frac{t}{2 \mu} r_w } = \mu + \frac{2 \mu}{t}
\end{align*}

Taking an upper bound of the approximation quality provided by each case, we get that the ring separator provides us a $\mu + 2 \frac{\mu^2}{t}$ rough approximation.
Substitute  $t \leq \frac{1}{\log n}$ and the time bound follows from the bound of the depth of the tree in Lemma \ref{assymringsep}.
\end{proof}

\begin{corollary}
 We can find a $O(\mu + 2 \mu^2 \log n)$ nearest neighbor to query point $q$ in $O( (\mu+1)^{2d} d^d \log n)$ time $\sqrt{\breg}$ using a $O(\frac{1}{\log n})$ ring-tree constructed in
 $O(d^d (\mu+1)^{2d} n \log(n))$ expected time.
\end{corollary}

\begin{proof}
 Set $t = \frac{1}{\log n}$, using Lemma \ref{assymringsep}. The construction time for the ring tree follows by combining Lemmas \ref{assymringsep}
and \ref{ringassym}.
\end{proof}

\subsection{Asymmetric quadtree decomposition}
\label{subsec:quadextension}

As in Section \ref{sec:finalized-algorithm}, we use the approximate near-neighbor returned by the ring-separator-tree query to progressively expand cells, using a subdivide-and-search procedure similar to Algorithm \ref{algo} albeit with 
$\sqrt{\sbreg}$ replaced with $\sqrt{\breg}$. A key difference is the procedure used to bisect a cell. 
%As before, we call the collection of all cells produced during the procedure our \emph{quadtree} . 

\begin{lemma}\label{bisectionasymmetric}
Let $G$ be a $\sqrt{\breg}$-box with maximum $\sqrt{\breg}$-side-length $s$ and $G_i$ its subcells 
produced by partitioning each side of $G$ into two equal intervals under $\sqrt{\breg}$. For all non-empty subboxes $G_i$ of $G$, we can find $p_i \in P \cap G_i$ in $O(2^d \log^{d} n)$ total time complexity, and the maximum $\sqrt{\breg}$-side-length of any $G_i$ is at most $\frac{s}{\sqrt{2}}$.
\end{lemma}
\begin{proof}
 Note that $G$ is defined as a product of $d$ intervals. For each interval, we can find an approximate bisecting point 
under $\sqrt{\breg}$ in $O(1)$ time. Here the bisection point $x$ of interval $[a b]$ is such that 
$\sqrt{\breg(a,x)} = \sqrt{\breg(x,b)}$. By resorting to the RTI for $\breg$, we get that $\breg(a,x) + \breg(x,b) < s^2$ and hence $\breg(a,x) = \breg(x,b) < \frac{s^2}{2}$ which implies $\sqrt{\breg(a,x)} = 
\sqrt{\breg(x,b)} < \frac{s}{\sqrt{2}}$. The rest of our proof follows as in Lemma \ref{bisectionProcedure}.

\begin{figure}[H]
  \begin{center}
    \includegraphics[scale = 0.7]{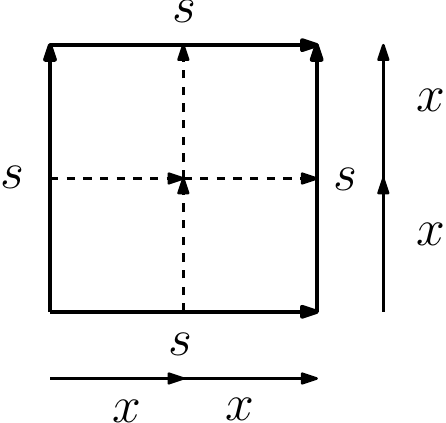}
  \end{center}
  \caption{Illustrating bisection of a box, in this case of equal side lengths $s$. The arrows show directionality. }
\end{figure}
\end{proof}

We now bound the number of cells that will be added to our search queue. 
We do so indirectly, by placing a lower bound on the maximum $\sqrt{\breg}$-side-length of all such cells, and note that for the
asymmetric case we get an additional factor of $\frac{1}{\mu+1}$.

\begin{lemma}\label{adiamtolength}
The $\sqrt{\breg}$-diameter of an $\sqrt{\breg}$-cube $C$ of $\sqrt{\breg}$-side-length $s$ is bounded by $(\mu+1) \sqrt{d} s$. 
\end{lemma}
\begin{proof}
Since the cube may be covered by a $\sqrt{\breg}$-left-ball of radius $\sqrt{d} s$ placed at a suitably chosen corner (by lemma \ref{acube}), lemma \ref{leftcircle} bounding the diameter of such a ball gives the required bound on the diameter
of the cube.
\end{proof}

\begin{lemma}\label{depthAsymmetricCube}
Algorithm \ref{algo} (with $\sqrt{\sbreg}$ replaced by $\sqrt{\breg}$) will not add the children of node  $\textbf{C}$ to our search queue 
if the maximum $\sqrt{\breg}$-side-length of $\textbf{C}$ is less than $\frac{\eps D({nn}_q,q)}{2 \mu (\mu+1) \sqrt{d}}$.
\end{lemma}
\begin{proof}
Let $\Delta(\textbf{C})$ represent the diameter or maximum distance between any two points of cell $\textbf{C}$.

By construction, we can expand $\textbf{C}$ only if some subcell of $\textbf{C}$ contains a point $p$ 
such that $\sqrt{\breg}(p , q) \leq (1 - \frac{\eps}{2}) D_{\text{near}}$. Note that since $\textbf{C}$ is examined, 
we have $D_{\text{near}} \leq p(\text{rep}_C , q)$.
 Now assuming we expand $\textbf{C}$, then we must have:
\begin{align*}
 \sqrt{\breg}(\text{rep}_C,q) - \sqrt{\breg}(p, q) &< \mu \Delta(\textbf{C}) \\
D_{\text{near}} - \left(1 - \frac{\eps}{2} \right)D_{\text{near}} &< \mu \Delta(\textbf{C}) \\ 
\frac{\eps}{2} D_{\text{near}} &< \mu \Delta(\textbf{C}) \\
\frac{\eps}{2 \mu} D_{\text{near}} &< \Delta(\textbf{C})
\end{align*}
Note that we substitute $\sqrt{\breg} (\text{rep}_C, q) < D_{\text{near}} $ and that by the definition of $D_{\text{near}}$ 
as our candidate nearest neighbor distance,  $\sqrt{\breg}( \text{nn}_q,q) < D_{\text{near}}$. 
Our main modification from the symmetric case is that here 
$\Delta(\textbf{C}) < (\mu+1)\sqrt{d} s$ by lemma \ref{adiamtolength},  where $s$ is the maximum side-length of $\textbf{C}$, 
as opposed to $\sqrt{d}s$. 
\end{proof}

The main difference between this lemma and Lemma~\ref{depthCube} is the extra factor of $\mu+1$ that we incur 
(as discussed) because of asymmetry. We only need do a little more work to obtain our final buonds:
\begin{lemma}\label{timeFinalAsymmetric}
Given a $\sqrt{\breg}$-cube $G$ of $\sqrt{\breg}$-side-length $D_{\text{rough}}$, and letting $x=  \frac{1}{\eps^d} 2^d \mu^d (\mu+1)^d d^{\frac{d}{2}} \left( \frac{D_{\text{rough}}}{\sqrt{\breg}(\text{nn}_q,q)} \right)^d $  we can compute a $(1 + \eps)$- right sided nearest neighbor to $q$ in $G$ 
in $O(x^2 \log^d n)$ time.
\end{lemma}

\begin{proof}
Consider a quadtree search from $q$ on a $\sqrt{\breg}$-cube $G$ of $\sqrt{\breg}$-side-length $D_{\text{rough}}$. By lemma \ref{depthAsymmetricCube},
 our algorithm will not expand cells with all $\sqrt{\breg}$-side-lengths smaller than  $\eps  \sqrt{\breg}( \text{nn}_q,q) / 2 \mu (\mu+1) \sqrt{d}$.
 But since the $\sqrt{\breg}$-side-length reduces by at least a factor of $\sqrt{2}$ in each dimension upon each split, 
all $\sqrt{\breg}$-side-lengths are less than this value after  $k = \log_{\sqrt{2}} \left( 2D_{\text{rough}}  \mu (\mu+1) \sqrt{d}/ \eps  \sqrt{\breg}(\text{nn}_q,q) \right)$ repeated bisections of our root cube. Observe now that $O(\log^d n)$ time is spent at each node by Lemma \ref{bisectionasymmetric}, that at the $k$-th level the 
number of nodes expanded is $2^{kd}$, and that $\log_{\sqrt{2}} n = (\log_{2} n )^2$.
We then get a final time complexity bound of 
$O\left( (1/\eps^{2d}) 2^{2d} \mu^{2d} (\mu+1)^{2d} d^d \left( D_{\text{rough}}/\sqrt{\breg}( \text{nn}_q, q) \right)^{2d} \log^d n \right)$. 
\qedhere
\end{proof}

Substituting $D_{\text{rough}} = \mu^2 \log (n) \sqrt{\breg}( \text{nn}_q,q) $ in Lemma \ref{timeFinalAsymmetric} gives us a  bound of $O\left(2^{2d} \frac{1}{\eps^{2d}}  \mu^{6d} (\mu+1)^{2d} d^d \log^{3d} n\right)$.  This time is per cube of $F$ that covers right-ball $B(q,D_{\text{rough}})$. Noting that there are $2^d$ such cubes gives us a final time complexity of $O\left(2^{3d}   \frac{1}{\eps^{2d}} \mu^{6d} (\mu+1)^{2d} d^d \log^{3d} n \right)$. The space bound follows as in Section~\ref{sec:finalized-algorithm}.

\paragraph{Logarithmic bounds for Asymmetric Bregman divergences}
\label{subsec:assymcondition}

We now extend our logarithmic bounds from Section \ref{sec:condition} 
to asymmetric Bregman divergence $\sqrt{\breg}$.
First note that the following lemma goes through by identical argument 
to lemma~\ref{EucBregBisect}.
\begin{lemma}\label{EucABregBisect}
Given an interval $I = [x_1 x_2] \subset \reals$ s.t. $x_1 < x_2$, $D_e(x_1, x_2) = r_e$ and $\sqrt{\breg}(x_1,x_2) = r_{\phi}$, suppose we divide $I$ into $m$ subintervals of equal length under $D_e$ with endpoints $x_1 = a_0 < a_1 < \ldots < a_{m-1} < a_m = x_2$ where $D_e (a_i, a_{i+1}) = r_e /m$, for all $i\in [0 \ldots m-1]$. Then $\frac{r_{\phi} }{c_0 m } \leq \sqrt{\breg} (a_i, a_{i+1}) \leq  \frac{c_0 r_{\phi}}{m}$.
\end{lemma}
\begin{corollary}\label{repabis}
If we recursively bisect an interval $I = [x_1 x_2] \subset \reals$ 
s.t. $D_e(x_1, x_2) = r_e$ and $\sqrt{\breg(x_1,x_2)} = r_{\phi}$ into $2^i$ equal subintervals (under $D_e$), then $\frac{r_{\phi} }{c_0 2^i } \leq \sqrt{\breg (a_k, a_{k+1}) }\leq  \frac{c_0 r_{\phi}}{2^i}$ for any of the subintervals $[a_k a_{k+1}]$ so obtained.
Hence after $i = \lceil \log \frac{c_0 r_{\phi}}{ x} \rceil$ subdivisions, all intervals will be of length at most $x$ under $\sqrt{\breg}$.
\end{corollary}

We now construct a compressed Euclidean quad tree as before, 
modifying the Section \ref{sec:condition} analysis slightly 
to account for the weaker packing bounds for $\sqrt{\breg}$ and the extra $\mu+1$ factor on the diameter of a cell. 

\begin{theorem}
  Given an asymmetric decomposable Bregman 
divergence $D_\phi$ that is $\mu$-defective over a domain  with associated 
$c_0$ as in Section \ref{sec:condition}, we can compute a $(1+\epsilon)$-approximate right-near-neighbor in 
time $O \left((\mu+1)^d d^{\frac{d}{2}} \log n + ( \frac{2{c_0}^4 ( \mu+1) \mu^3 \sqrt{d}}{\eps})^d \right)$.
\end{theorem}

We note our first new Lemma, a slightly modified packing bound due to $\sqrt{\breg}$ not having a direct RTI.

\begin{lemma}\label{cpack}
 Given an interval $[x_1 x_2] \subset \reals$ s.t. $\sqrt{\breg(x_1,x_2)} = r > 0$, and 
intervals with endpoints $a_0 < a_1 < \ldots < a_{m-1} < a_m$, s.t. for all $i\in[0 \ldots m-1]$, $\sqrt{\breg(a_i, a_{i+1})} \geq l$, 
at most $O(\frac{c_0 r}{l})$ such intervals intersect $[x_1 x_2]$.
\end{lemma}

\begin{proof}
By the Lagrange form, 
\begin{equation}
\frac{l}{r} < \frac{\sqrt{\breg(a_i, a_{i+1})}}{\sqrt{\breg(x1,x2)}} 
< c_0 \frac{D_e(a_i,a_{i+1})}{D_e(x_1,x_2)},
\end{equation}
or we can say that $\frac{D_e(a_i,a_{i+1})}{D_e(x1,x2)} > \frac{l}{r c_0}$.
The RTI for $D_e$ then gives us the required result.
\end{proof}

\begin{corollary}\label{cpackd}
 Given a ball $B$ of radius $r$ under $\sqrt{\breg}$, there can be at most
$c_0^d (\frac{r}{l})^d$ disjoint $\sqrt{\breg}$-cubes that can intersect $B$ where each cube has 
side-length at least $l$ under $\sqrt{\breg}$.
\end{corollary}

 As before, we  find the smallest enclosing Bregman cube of side-length $s$
that encloses our point set, and then construct a compressed Euclidean quad-tree in
pre-processing. Let $L_i$ denote the cells at the $i$-th level. 

 \begin{lemma}\label{apackingTheTree}
 Given a $\sqrt{\breg}$ right-ball $B$ of radius $r$ under $\sqrt{\breg}$, let $i = \log \frac{ s}{c_0 r}$. 
Then $|L_i \cap B| \leq O({c_0}^d)$ and the side-lengths of each cell in $L_i$ are
between $r$ and ${c_0}^2 r$ under $\sqrt{\sbreg}$. We can also explicitly retrieve 
the quadtree cells corresponding to $|L_i \cap B|$ in $O({c_0}^d \log n)$ time.
 \end{lemma}

 \begin{proof}
 Note that for cells in $L_i$, we have $\sqrt{\breg}$-side-lengths between 
$\frac{s}{c_0 2^i}$ and $\frac{c_0 s}{2^i}$ by Corollary~\ref{repabis}.
  Substituting $i = \log \frac{s}{c_0 r}$,  these cells have side-length 
between $r$ and ${c_0}^2 r$ under $\sqrt{\sbreg}$.
 Now, we look in each dimension at the number of disjoint intervals of length 
at least $r$ that can intersect $B$. By Lemma~\ref{cpack}, this is at most $c_0$.
 The rest of the proof follows as in Lemma~\ref{packingTheTree}. 
 \end{proof}

 We first obtain in $O(\log n)$ time with our asymmetric ring-tree an $O(n)$ 
ANN $q_{\text{rough}}$ to query point $q$, such that
$D_{\text{rough}} = \sqrt{\breg}( q_{\text{rough}},q) = O \left(\mu^2 n \sqrt{\breg} ( \text{nn}_q,q) \right)$.
 We then use Lemma \ref{apackingTheTree} to get $O({c_0}^d)$ cells of our quadtree that 
intersect right-ball $B \left(q, \sqrt{\breg ( q_{\text{rough}},q)} \right)$. 

Let us call this collection of cells as $Q$. We then carry out a quadtree search on each element
of $Q$. Note that we expand only cells which may contain a point nearer to query point $q$
than the current best candidate. We bound the depth of our search using $\mu$-defectiveness
similar to Lemma~\ref{treeDepth}.

\begin{lemma}\label{modADepth}
 We need only expand cells of $\sqrt{\breg}$-diameter greater than 
$\frac{\eps  \sqrt{\breg( \text{nn}_q,q)}}{2 \mu}$
\end{lemma}

\begin{proof}
 By $\mu$-defectiveness, similar to Lemma~\ref{depthCube}.
\end{proof}

\begin{corollary}\label{asidelen}
 We will not expand cells where the length of each side is less than 
$x =\frac{\eps \sqrt{ \breg( \text{nn}_q,q)}}{2 \mu (\mu+1) \sqrt{d}}$ under $\sqrt{\breg}$.
\end{corollary}

 \begin{proof}
Note that a quadtree cell $\textbf{C}$ where every side-length is less than $x$ can be covered by a ball of radius $\sqrt{d} x$ under 
 $\sqrt{\breg}$ with appropriately chosen corner as center of ball, 
as explained in proof of Lemma \ref{acube}. Now by Lemma~\ref{leftcircle}, 
$\sqrt{\breg(a,b)} \leq (\mu+1) \sqrt{d} x$, $\forall a,b \in \textbf{C}$. 
Substituting for $x$ from Lemma \ref{modADepth}, the $\sqrt{\breg}$-diameter of $\textbf{C}$ is at 
most $\frac{\eps \sqrt{\breg}(\text{nn}_q,q)}{2 \mu}$.
 \end{proof}

Let the spread be $\beta = \frac{D_{\text{rough}}}{\sqrt{\breg}(\text{nn}_q, q)} = O(\mu^2 n)$.
\begin{lemma}\label{AfinalTreeDepth}
We will only expand our tree to a depth of 
$k = \log(2 c_0^3 \mu (\mu+1) \beta \sqrt{d} / \eps)$.
\end{lemma}

\begin{proof}
 Note first that $D_{\text{rough}}= O \left(\beta \sqrt{\breg} ( \text{nn}_q,q) \right)$. Then 
by Lemma \ref{apackingTheTree}, each of the cells of our corresponding quadtree is of $\sqrt{\breg}$-side-length at most
$c_0^2 D_{\text{rough}}$.  Using \ref{asidelen} to lower bound the minimum $\sqrt{\breg}$-side-length of any quadtree cell expanded, and \ref{repabis} to bound number of bisections needed to guarantee all $\sqrt{\breg}$-side lengths are within this  gives us out bound.
\end{proof}

\begin{lemma}\label{atreebreadth}
 The number of cells expanded at the $i$-th level is 
$n_i < 2^d(\mu^d d^{\frac{d}{2}} c_0^{5d} + (\frac{c_0^2 2^i}{\beta})^d)$.
\end{lemma}
 
\begin{proof}
Recalling that the cells of $Q$ start with all $\sqrt{\breg}$-side-lengths at most $c_0^2 D_{\text{rough}}$,
at the $i$-th level the side-length of a cell $\textbf{C}$ is at most $\frac{c_0^3  
D_{\text{rough}}}{2^i}$ under $\sqrt{\breg}$ by Corollary~\ref{repabis}. 
And using Lemma~\ref{leftcircle},  $\Delta{\textbf{C}} < \sqrt{d}(\mu+1) \frac{c_0^3  D_{\text{rough}}}{2^i}$.
Hence by $\mu$-defectiveness there must be a point at distance at most 
$D_{\text{best}} = \sqrt{\breg(\text{nn}_q,q)} + \frac{\mu (\mu+1) c_0^3 \sqrt{d} D_{\text{rough}}}{2^i}$.

The $\sqrt{\breg}$-side-length of a cell $C$ at this level is at least $\frac{D_{\text{rough}}}{c_0 2^i}$, so the number of cells expanded is at most $n_i = c_0^d(\frac{D_{\text{best}}}{\Delta{\textbf{c}}})^d = 
c_0^d( \mu (\mu+1) \sqrt{d} c_0^4   + \frac{c_0 2^i}{\beta})^d$, 
by Corollary \ref{cpackd}.
Using the fact that $(a+b)^d < 2^d (a^d + b^d)$, we get 
$n_i < 2^d \left(\mu^d (\mu+1)^d d^{\frac{d}{2}} c_0^{5d} + (\frac{c_0^2 2^i}{\beta})^d \right)$.
\end{proof}

Simply summing up all $i$, the total number of nodes explored is
\[
O(2^d \mu^d (\mu+1)^d c_0^{5d} \log(2 c_0^3 \mu \beta \sqrt{d} /\eps) + 2^{2d} c_0^{5d} \mu^d
(\mu+1)^d d^{\frac{d}{2}}/\eps^d),
\]
or
\[
O \left(2^d \mu^d (\mu+1)^d c_0^{5d} \log n + 2^{2d} c_0^{5d} \mu^d
(\mu+1)^d d^{\frac{d}{2}}/\eps^d \right),
\]
after substituting back for $\beta$ and ignoring smaller terms.
Recalling that there are $c_0^d$ cells in $Q$ adds a further $c_0^d$ multiplicative factor.
This time complexity of this quadtree phase(number of cells explored) of our algorithm dominates the time complexity of the ring-tree 
search phase of our algorithm, and hence is our overall time
complexity for finding a $(1+ \eps)$ ANN to $q$.
For space and pre-construction time, we note that compressed Euclidean quadtrees can be built in $O(n \log n)$ time and require $O(n)$ space~\cite{snotes},
which matches our bound for the ring-tree construction phase of our algorithm requiring $O(n \log n)$ time and $O(n)$ space. 

\section{Numerical arguments for bisection}\label{sec:numerical}
In our algorithms, we are required to \emph{bisect} a given interval with respect to the distance measure $D$, as well as construct points that lie a fixed distance away from a given point. We note that in both these operations, we do not need exact answers: a constant factor approximation suffices to preserve all asymptotic bounds.  
In particular, our algorithms assume two procedures: 
\begin{enumerate}
\item{Given interval $[ab] \subset \reals$, find $\bar{x} \in [ab]$ such that $(1 - \alpha)\sqrt{\sbreg(a,\bar{x})} < 
\sqrt{\sbreg(\bar{x}, b)} < (1 + \alpha) \sqrt{\sbreg(a,\bar{x})}$}

\item{ Given $q \in \reals$ and distance $r$, find $\bar{x}$ s.t $|\sqrt{\sbreg (q,\bar{x})} - r| < \alpha r$} 
\end{enumerate}

Cayton presents a similar bisection procedure~\cite{caytonpaper} as ours for the second task above, although our analysis of
the convergence time is more explicit in our parameters of $\mu$ and $c_0$. For a given $\sqrt{\sbreg} : \reals \to \reals$ and precision parameter $0 < \alpha < 1$, we describe a procedure that 
yields an $0 <\alpha < 1$ approximation in $O(\log c_0 + \log \mu + \log \frac{1}{\alpha})$ steps for both problems, where $c_0$ implicitly depends on the domain of convex function $\phi$:

\begin{equation}
c_0 =\sqrt{\max_{ 1 \leq i \leq d} \left(\max_x \phi_i'' (x)/  \min_y \phi_i''(y) \right)}
\end{equation}

Note that this implies linear convergence. While more involved numerical methods such as Newton's method may yield better results, 
our approximation algorithm serve as proof-of-concept that the numerical precision is not problematic. 

A careful adjustment of our NN-analysis now gives  
a $O\left( \left(\log \mu + \log c_0  + \log \frac{1}{\alpha} \right )2^{2d}  (1+ \alpha)^d \frac{1}{\eps^d} \mu^{3d} d^{\frac{d}{2}} \log^{2d} n \right)$ time complexity to compute a $(1 + \eps)$-ANN to query point $q$.

We now describe some useful properties of $\sbreg$.

\begin{lemma}\label{ratio}
Consider  $\sqrt{\sbreg} : \reals \to \reals$ such that $c_0 = \sqrt{\max_x \phi''(x) / \min_y \phi''(y)}$. Then for any two intervals $[x_1 x_2] ,[x_3 x_4] \subset \reals$ , 

\begin{equation}
\frac{1}{c_0} \frac{|x_1 - x_2| }{ |x_3 - x_4|} < \frac{\sqrt{\sbreg(x_1, x_2)} }{ \sqrt{\sbreg(x_3,x_4)}} < c_0 
\frac{ |x_1 - x_2| }{ |x_3 - x_4|} 
\end{equation}

\begin{proof}
The lemma follows by the definition of $c_0$ and by direct computation from the Lagrange form of $\sqrt{\sbreg(a,b)}$, i.e, $\sqrt{\sbreg(a,b)} = \sqrt{\phi''(\bar{x}_{ab})} |b -a|$, for some 
$\bar {x}_{ab} \in [ab]$.
\end{proof}

\end{lemma}

\begin{lemma}\label{approxnum}
Given a point $q \in \reals$,  distance $r \in \reals$, precision parameter $0 <\alpha < 1$ and a $\mu$-defective $\sqrt{\sbreg} : \reals \to \reals$, we can locate a  point $x_i$ such that $|\sqrt{\sbreg(q,x_i)} - r| <  \alpha r$ in  $O(\log \frac{1}{\alpha} + \log \mu +  \log c_0)$ time.
\end{lemma}

\begin{proof} 
Let $x$ be the point such that $\sqrt{\sbreg(q,x)} = r$. We outline an iterative process,  \ref{finalalgo}, with $i$-th iterate $x_i$ that converges to $x$.
\begin{algorithm}
  \caption{QueryApproxDist$(q,r,c_0, \alpha)$}
  \begin{algorithmic}
    \STATE Let $x_0 > q$ be such that $\frac{\sqrt{\phi''(q)}}{c_0} (x_0 - q) = r$
	\STATE Let $\text{step} = (x_0 - q)/2$

    \REPEAT 
    	\IF{$\sqrt{\sbreg(q,x_i)} < r$}
			\STATE	$x_{i+1} = x_i + \text{step}$
		\ELSE 
			\STATE $ x_{i+1} = x_i- \text{step}$		
    	\ENDIF
		
		\STATE $\text{step} = \text{step}/2$
    \UNTIL {$|\sqrt{\sbreg(q,x_i)} - r| \leq \alpha r$}
    \STATE Return $\bar{x} = x_i$
  \end{algorithmic}
  \label{finalalgo}
\end{algorithm}
First note that $\frac{\sqrt{ \phi''(q)}}{c_0} \leq \sqrt{\min_y \phi''(y)}$ and 
$\frac{\sqrt{\phi''(q)}}{ c_0} \geq  \frac{\max_z \sqrt{\phi''(z)}} {c_0^2}$. It immediately follows that $r \leq\sqrt{\sbreg(q,x_0)} \leq c_0^2 r$.

By construction, $|x_i - x| \leq |x_0 - q|/2^i$. Hence by Lemma \ref{ratio}, $\sqrt{\sbreg(x_i,x)} < \frac {c_0 ^3 r}{2^i}$.  We now use $\mu$-defectiveness to upper bound our error $|\sqrt{\sbreg (q,x_i)} - \sqrt{\sbreg (q,x)}|$ at the $i$-th iteration:

\begin{equation}
\left|\sqrt{\sbreg (q,x_i)} - \sqrt{\sbreg (q,x)} \right | < \frac{\mu c_0^3 r}{2^i} 
\end{equation} 

Choosing $i$ such that $(\mu c_0^3)/2^i \leq \alpha$ implies that $i \leq \log \frac{1}{\alpha} +  \log \mu + 3 \log c_0$.\qedhere
\end{proof}
An almost identical procedure can locate an  approximate bisection point of interval $[ab]$ in $O(\log \mu + \log c_0 + \log \frac{1}{\alpha})$ time, and similar techniques can be applied for $\sqrt{\breg}$. We omit the details here.

\section{Further work}
A major open question is whether bounds independent of $\mu$-defectiveness can be obtained for the complexity of ANN-search under Bregman divergences. As we have seen, traditional grid based methods rely heavily on the triangle inequality and packing bounds, and there are technical difficulties in adapting other method such as cone decompositions~\cite{chanNN} or approximate Voronoi diagrams~\cite{plebs}. We expect that we will need to exploit geometry of Bregman divergences more substantially. 
\section{Acknowledgements}
We thank Sariel Har-Peled and anonymous reviewers for helpful comments.

\appendix
\section{Proofs from Section \ref{sec:properties}}
\label{sec:bounded}
\begin{lemma}\label{app:1}
Given any interval $I=[x_1 x_2]$ on the real line, there exists a finite $\mu$ such that  $\sqrt{\sbreg}$ is $\mu$-defective with respect to $I$. We require all order derivatives of $\phi$ to be defined and bounded over the closure of $I$,
and $\phi''$ to be bounded away from zero.
\end{lemma}

\begin{proof}
Consider three points $a,b,q \in I$. 

Due to symmetry of the cases and conditions, there are three cases to consider: $a<q<b$, $a < b <q$ and $q < b < a$.
\begin{description}
\item[\textbf{Case 1:}]
Here $a < q < b$. The following is trivially true by the monotonicity of $\sqrt{\sbreg}$.
\begin{equation}
\left| \sqrt{\sbreg(q,a)} - \sqrt{\sbreg(q,b)} \right| <  \sqrt{\sbreg(a,b)}
\end{equation}

\item[\textbf{Cases 2 and 3:}]
For the remaining symmetric cases, $a < b< q$ and $q<b<a$, note that since $\sqrt{\sbreg(q,a)} - \sqrt{\sbreg(q,b)}$ and $\sqrt{\sbreg(a,b)}$ are both bounded, continuous functions on a compact domain (the interval $[x_1 x_2]$), we need only show that the following limit exists:
\begin{equation}\label{mu-def-limit}
  \lim_{a \to b} \frac{\left| \sqrt{\sbreg(q,a)} - \sqrt{\sbreg(q,b)} \right|}{\sqrt{\sbreg(a,b)}}
\end{equation}

First consider $a<b<q$, and we assume $\lim_{b \to a}$. For ease of computation, we replace $\phi'$ by $\psi$, to be restored 
at the last step. We will use the following Taylor expansions repeatedly in our derivation: $b = a +h$, $ \psi(b) = \psi(a+h) =  \psi(a) + h\psi'(a) + E(h^2)$,  and $ \sqrt{1+h} =  1 + h/2 + E(h^2)$.  Here $E(h^x)$ denotes a tail of a Taylor expansion
around $a$ (or equivalently a Maclaurin expansion in $h$) where the lowest order term is $h^x$. Since we will be handling
multiple Taylor expansions in what follows, we will use subscripts of the form $E_1$, $E_2$, etc. to distinguish the tails
of different series.  
\begin{equation}\label{mucomp}
 \frac{\sqrt{\sbreg(a,q)} - \sqrt{\sbreg(b,q)}}{\sqrt{\sbreg(a,b)}} 
= \frac{ \left(\sqrt{(q-a)(\psi(q) - \psi(a)) }  - \sqrt{(q-b)(\psi(q) - \psi(b)) } \right)}
       {\sqrt{(b-a)(\psi(b) - \psi(a)) }}
\end{equation}

Computing the denominator, using the expansion that $\psi(b) = \psi(a+h) = \psi(a) + h \psi'(a) + E_1(h^2)$, we get:
\begin{align*}
& \sqrt{(b-a) (\psi(b) - \psi(a)) }  
\\ =& \sqrt{h (\psi(a+h) - \psi(a)) } 
\\ =& \sqrt{h ( \psi(a) + h \psi'(a) + E_1(h^2) - \psi(a)) } 
\\ =& \sqrt{h (h \psi'(a)+ E_1(h^2))} 
\\ =& \sqrt{h (h \psi'(a)+ h E_2(h))} 
\\ =&  \sqrt{h^2(\psi'(a)+ E_2(h))}
\\ =& h \sqrt{\psi'(a) + E_2(h)} 
\end{align*}
Where in the third last step we set $h E_2(h) = E_1(h^2)$.

We now address the numerator, and begin by taking the same Taylor expansion.
\begin{align*}
&\sqrt{(q-a)(\psi(q) - \psi(a))} - \sqrt{(q-b) (\psi(q) - \psi(b)) } 
\\ =&  \sqrt{(q-a) (\psi(q) - \psi(a))}  - \sqrt{(q-a-h) (\psi(q) - \psi(a) - h \psi'(a)-E_1(h^2)) } 
\\ =&  \sqrt{(q-a) (\psi(q) - \psi(a))}  - \sqrt{(q-a)\left( 1-\frac{h}{q-a} \right) (\psi(q) - \psi(a)) \left( 1 -  \frac{h \psi'(a) + E_1(h^2)}{\psi(q) - \psi(a)} \right)}
\\ =& \sqrt{(\psi(q) - \psi(a))(q-a)} \left(1 - \sqrt{ 1 - \frac{h}{q-a}} \sqrt{ 1 - \frac{h \psi'(a) + E_1(h^2)}{\psi(q) - \psi(a)} } \right) 
\end{align*}

We now take the McLaurin expansion of the square roots and note for the second such expansion we gain more higher order terms of $h$ which we merge with
$E_1$ to obtain a $E_3(h^2)$. 
\begin{align*}
&\sqrt{(q-a)(\psi(q) - \psi(a))} - \sqrt{(q-b) (\psi(q) - \psi(b)) } 
\\ =&  \sqrt{(\psi(q) - \psi(a))(q-a)} \left(1 - \left( 1- \frac{h}{2(q-a)} +E_4(h^2)\right) \left(1 -  \frac{h \psi'(a)}{2(\psi(q) - \psi(a))} + E_3(h^2) \right) \right) 
\\ =& \sqrt{(\psi(q) - \psi(a))(q-a)} \left( \frac{h}{2(q-a)} + h \frac{\psi'(a)}{2(\psi(q) - \psi(a))}+E_5(h^2) \right) 
\\ =& h\sqrt{(\psi(q) - \psi(a))(q-a)} \left( \frac{1}{2(q-a)} +  \frac{\psi'(a)}{2(\psi(q) - \psi(a))}+E_6(h) \right) 
\end{align*}

Where the $E_5(h^2)$ is again obtained in the second last step from merging products involving one of $E_3(h^2)$ or $E_4(h^2)$ , as well as of the two terms involving $h$ with each other. And in the last step, we set $E_5(h^2) = h E_6(h)$.

Now combine numerator and denominator back in equation \ref{mucomp} and observe a factor of $h$ cancels out. 
\begin{align*}
&\frac{\sqrt{\sbreg(a,q)} - \sqrt{\sbreg(b,q)}}{\sqrt{\sbreg(a,b)}} \\
=& \frac{ h \sqrt{(\psi(q) - \psi(a))(q-a)} 
                         \left( \frac{1}{2 (q-a)} + \frac{\psi'(a)}{2 (\psi(q) - \psi(a))} + E_6(h) \right)}
        { h \sqrt{\psi'(a) + E_2(h)}} 
\\ = &\sqrt{\frac{(\psi(q) - \psi(a))(q-a)}{\psi'(a)  + E_2(h)}} 
      \left( \frac{1}{2 (q-a)} + \frac{\psi'(a)}{2 (\psi(q) - \psi(a))} + E_6(h) \right) 
\end{align*}

Note now that since $\phi$ is strictly convex, neither the numerator nor denominator of this expression approach $0$ as 
$\lim_{h \to 0}$ (or equivalently, $\lim_{b \to a}$). So we can safely drop the higher order terms in the limit to obtain:
\begin{align*}\lim_{h \to 0} &\frac{\sqrt{\sbreg(a,q)} - \sqrt{\sbreg(b,q)}}{\sqrt{\sbreg(a,b)}} 
\\=& \sqrt{\frac{(\psi(q) - \psi(a))(q-a)}{\psi'(a)}} 
      \left( \frac{1}{2 (q-a)} + \frac{\psi'(a)}{2 (\psi(q) - \psi(a))} \right)
\\=& \frac{1}{2}\sqrt{\frac{(\psi(q) - \psi(a))(q-a)}{\psi'(a)}} 
      \left( \frac{1}{ \sqrt{q-a} \sqrt{q-a}} + \frac{\sqrt{\psi'(a)} \sqrt{\psi'(a)}}{ \sqrt{\psi(q) - \psi(a)}
      \sqrt{\psi(q) - \psi(a)}} \right)      
\\=& \frac{1}{2} \left( \sqrt{\frac{\psi(q) - \psi(a)}{\psi'(a)(q-a)}} + \sqrt{\frac{\psi'(a)(q-a)}{\psi(q) - \psi(a)}} \right) \end{align*}

Substituting back $\phi'(x)$ for $\psi(x)$, we see that limit \ref{mu-def-limit} exists, provided $\phi$ is strictly convex:

\begin{equation}
\frac{1}{2} \left( \sqrt{\frac{\phi'(q) - \phi'(a)}{\phi''(a)(q-a)}} 
                  + \sqrt{\frac{\phi''(a)(q-a)}{\phi'(q) - \phi'(a)}} \right)
\end{equation}

\end{description}
The analysis follows symmetrically for case 3, where $q < b < a$.\qedhere
\end{proof}

\begin{lemma}\label{app:2}
Given any interval $I=[x_1 x_2]$ on the real line, there exists a finite $\mu$ such that  $\sqrt{\breg}$ is right-sided $\mu$-defective with respect to $I$.We require all order derivatives of $\phi$ to be defined and bounded over the closure of $I$,
and $\phi''$ to be bounded away from zero. 
\end{lemma}

\begin{proof}
Consider any three points $a,b,q \in I$. We will prove that there exists finite $\mu$ such that:
\begin{equation}
\left| \sqrt{\breg(a,q)} - \sqrt{\breg(b,q)} \right| <  \mu \sqrt{\breg(b,a)}
\end{equation}

 Here there are now six cases to consider: $a<q<b$, $b<q<a$, $a < b <q$, $b < a <q$, $q < b < a$, and $q<a<b$ .

\begin{description}
\item[\textbf{Case 1 and 2:}]
Here $a < q < b$.  By monotonicity we have that:
\begin{equation}
\left| \sqrt{\breg(a,q)} - \sqrt{\breg(b,q)} \right| <   \sqrt{\breg(a,b)}  + \sqrt{\breg(b,a)}
\end{equation}

But by lemma \ref{firstTosecond}, we have that $\sqrt{\breg(a,b)} < c \sqrt{\breg(b,a)}$ for some parameter $c_0$ defined over $I$.
This implies that $ \left| \sqrt{\breg(a,q)} - \sqrt{\breg(b,q)} \right| / \sqrt{\breg(b,a)} < c_0 + 1$, i.e, it is bounded over $I$. A similar analysis works for Case 2 where $b<q<a$.

\item[\textbf{Cases 3 and 4:}]
For these two cases, $a < b< q$ and $b < a < q$, note that since $\sqrt{\breg(q,a)} - \sqrt{\breg(q,b)}$ and $\sqrt{\breg(b,a)}$ are both bounded, continuous functions on a compact domain (the interval $[x_1 x_2]$), we need only show that the following limit exists:
\begin{equation}\label{mubreg-def-limit}
  \lim_{a \to b} \frac{\left| \sqrt{\breg(a,q)} - \sqrt{\breg(b,q)} \right|}{\sqrt{\breg(b,a)}}
\end{equation}

First consider $a<b<q$, and we assume $\lim_{b \to a}$.  For ease of computation, we replace $\phi'$ by $\psi$, to be restored 
at the last step. We will use the following Taylor expansions repeatedly in our derivation: $b = a +h$, $ \phi(b) = \phi(a+h) =  \phi(a) + h\phi'(a) + E(h^2)$, $ \phi(b) = \phi(a) + h \psi(a) + \frac{h^2 \psi'(a)}{2} + E(h^3)$ and $ \sqrt{1+h} =  1 + h/2 + E(h^2)$.  Here $E(h^x)$ denotes a tail of a Taylor expansion where the lowest order term is $h^x$. Since we will be handling
multiple Taylor expansions in what follows, we will use subscripts of the form $E_1$, $E_2$, etc. to distinguish the tails
of different series.  
\begin{equation}\label{mucom}
\lim_{a \to b} \frac{\sqrt{\breg(a,q)} - \sqrt{\breg(b,q)}}{\sqrt{\breg(b,a)}} 
= \lim_{a \to b} \frac{ \sqrt{\phi(a) - \phi(q) - \psi(q)(a-q) }  - \sqrt{\phi(b) - \phi(q) - \psi(q)(b-q)}}
       {\sqrt{ \phi(b) - \phi(a) - \psi(a)(b-a)}}
\end{equation}

Computing the denominator by replacing $b-a$ with $h$ and taking the Taylor expansion of $\phi(b)$:
\begin{align*}
	\sqrt{ \phi(b) - \phi(a) - \psi(a)(b-a)}
   &= \sqrt{\left(\phi(a) + h \psi(a) + \frac{h^2 \psi'(a)}{2} +E_1(h^3)\right) - \phi(a) -  \psi(a) h } 
\\ &= \sqrt{\frac{h^2 \psi'(a)}{2} + E_1(h^3)}  
\\ &= h \sqrt{\frac{\psi'(a)}{2} + E_2(h^2)} 
\end{align*}
Where in the last step, we let $E_1(h^3) = h E_2(h^2)$.
We now address the numerator:
\begin{align*}
 &\left(\sqrt{\phi(a) - \phi(q) - \psi(q)(a-q) }  - \sqrt{\phi(b) - \phi(q) - \psi(q)(b-q)} \right) 
 \\ =&  \sqrt{\phi(a) - \phi(q) - \psi(q)(a-q) }  - \sqrt{\phi(b) - \phi(q) - \psi(q)(b-a + a -q)}
 \\ =& \sqrt{\phi(a) - \phi(q) - \psi(q)(a-q) }  - \sqrt{\phi(b) - \phi(q) - \psi(q)(h + a -q)}
 \\  =& \sqrt{\phi(a) - \phi(q) - \psi(q)(a-q) }  - \sqrt{\phi(a) + h \psi(a) + E_3(h^2) - \phi(q) - \psi(q)(h + a -q)}.
\end{align*}
Where in the last step we took the Taylor Expansion of $\phi(b)$. Collecting terms of $h$ and continuing, we obtain:
\begin{align*} 
& \sqrt{\phi(a) - \phi(q) - \psi(q)(a-q) }  - \sqrt{\phi(a) - \phi(q) - \psi(q)(a -q) + h(\psi(a) - \psi(q)) + E_3(h^2)} 
\\ =& \sqrt{\breg(a,q)} - \sqrt{\breg(a,q) \left(1 + \frac{h(\psi(a) - \psi(q)) + E_3(h^2)}{\breg(a,q)} \right)}
\\ =&  \sqrt{\breg(a,q)} \left(1 - \sqrt{ 1 - \frac{h(\psi(q) - \psi(a))- E_3(h^2)}{\breg(a,q)}} \right) 
\\ =&  \sqrt{\breg(a,q)} \left(1 -  \left(1 - \frac{h(\psi(q) - \psi(a))}{2 \breg(a,q)} + E_4(h^2)\right) \right)
\\ =&  \frac{h \left(\psi(q) - \psi(a) - E_4(h^2)\right)}{2 \sqrt{\breg(a,q)}}. 
\end{align*}
Where we note in the above that the new error term of $E_4(h^2)$ was produced by combining $E_3(h^2)$ with the error term
produced by taking the Maclaurin expansion of the square root.

Now combine numerator and denominator back in equation \ref{mucom} and cancel a factor of $h$ accordingly, we get:
\begin{align*}  
&\frac{ \sqrt{\phi(a) - \phi(q) - \psi(q)(a-q) }  - \sqrt{\phi(b) - \phi(q) - \psi(q)(b-q)}}{\sqrt{ \phi(b) - \phi(a) - \psi(a)(b-a)}} \\
=& \left(\frac{h \left(\psi(q) - \psi(a) - E_4(h^2)\right)}{2 \sqrt{\breg(a,q)}} \right) / \left(h \sqrt{\frac{\psi'(a)}{2} + E_2(h^2)} \right) \\
=&  \left(\frac{ \left(\psi(q) - \psi(a) - E_4(h^2)\right)}{2 \sqrt{\breg(a,q)}} \right) / \left( \sqrt{\frac{\psi'(a)}{2} + E_2(h^2)} \right)
\end{align*}

Now if we take $\lim_{h \to 0}$ or equivalent $\lim_{a \to b}$, neither the numerator nor denominator of this
new expression become $0$ and indeed we may drop the higher order terms of $h$ safely.  Noting that $\breg(a,q) = \frac{1}{2}(\psi'(x))(q-a)^2$, for some $x \in [ab]$. 
\begin{align*}
\lim_{a \to b} \frac{\sqrt{\breg(a,q)} - \sqrt{\breg(b,q)}}{\sqrt{\breg(b,a)}} 
&= \frac{ 
                       \frac{(\psi(q) - \psi(a))}{2 \sqrt{\breg(a,q)}}  }
        {  \sqrt{\frac{\psi'(a)}{2}} } 
\\ &= \frac{(\psi(q) - \psi(a))}{q - a}\frac{\sqrt{\psi'(a)}}{\sqrt{ \psi'(x)}} 
\end{align*}

Substituting back $\phi'(x)$ for $\psi(x)$, we see that limit \ref{mubreg-def-limit} exists, provided $\phi$ is strictly convex:

\begin{equation}
 \frac{(\phi'(q) - \phi'(a))}{q - a}\frac{\sqrt{\phi''(a)}}{\sqrt{ \phi''(x)}} 
\end{equation}

The analysis follows symmetrically for case 4, by noting that $\lim_{a \to b} \frac{\sqrt{\breg (a,b)}}{\sqrt{\breg(b,a)}} = 1$ and that $\sqrt{\breg(a,q)} - \sqrt{\breg(b,q)}  = - (\sqrt{\breg(b,q)} - \sqrt{\breg(a,q)})$, i.e we may suitably interchange $a$ and $b$.

\item[\textbf{Cases 5 and 6:}]
Here $q<a<b$ or $q<b<a$. Looking more carefully at the analysis for cases 3 and 4, note that the ordering $q<a<b$ vs $a<b<q$ does not affect the magnitude of the expression for limit  \ref{mubreg-def-limit}, only the sign. Hence we can use the same analysis to prove $\mu$-defectiveness for cases 5 and 6.
\qedhere
\end{description}
\end{proof}

\begin{corollary}
 Given any interval $I=[x_1 x_2]$ on the real line, there exists a finite 
$\mu$ such that  $\sqrt{\breg}$ is left-sided $\mu$-defective with respect to $I$ 
\end{corollary}

\begin{proof}
Follows from similar computation. 
\end{proof}

\section{Discussion of empirical values of $\mu$.}
\label{sec:muranges}
We calculate now the values of $\mu$ observed for a selection of Bregman divergences points spread over a range of intervals; namely $[0.1 0.9]$, $[0.01 0.99]$ and
$[0.001 0.999]$.   Note that each of the values below is for the square root of the relevant divergence and that for the
Itakura Saito, Kullback-Liebler and Symmetrized Kullback-Liebler, $0$ is a boundary point where distances approach infinity. 
Interestingly, lemma \ref{AallDmusim} implies that whatever bounds for $\mu$ hold for points spread on an interval $I \in \reals$ also hold for points in the box $\prod_{i=1}^{d} I^d \in \reals^d$. We observe that for reasonable spreads of points,
while $\mu$ is not necessarily always small, it is also not a galactic constant as well.

\begin{center}
  \begin{tabular}{| l | c | r |}
    \hline
    Name & Interval Range & Value of $\mu$ \\ \hline
     & [0.1 0.9] & 2.35 \\ 
     Itakura-Saito & [0.01 0.99] & 7.17 \\ 
     & [0.001 0.999] & 22.42 \\   
     \hline
     & [0.1 0.9] & 1.65 \\ 
     Kullback-Liebler & [0.01 0.99] & 3.67 \\ 
     & [0.001 0.999] & 9.18 \\   
     \hline
     & [0.1 0.9] & 1.22 \\ 
     Symmetrized Kullback-Liebler & [0.01 0.99] & 2.42 \\ 
     & [0.001 0.999] & 6.05 \\   
     \hline
     & [0.1 0.9] & 1.14 \\ 
     Exponential  & [0.001 0.999] & 1.18 \\ 
     & [0.001 100] & 9.95 \\   
     \hline
  \end{tabular}
\end{center}
\newpage
\bibliography{nn}
\bibliographystyle{acm}
\end{document}